\documentclass[envcountsect,envcountsame, runningheads, a4paper, 10pt]{llncs}
\usepackage[utf8]{inputenc}
\usepackage{amsmath}
\usepackage[linesnumbered,ruled,vlined]{algorithm2e}
\let\oldnl\nl% Store \nl in \oldnl
\newcommand{\nonl}{\renewcommand{\nl}{\let\nl\oldnl}}% Remove line number for one line
\usepackage{url} 
\usepackage[normalem]{ulem}
\usepackage[hidelinks,breaklinks=true]{hyperref}
\hypersetup{colorlinks,linkcolor={black},citecolor={blue},urlcolor={black}}
\usepackage[square,sort,comma,numbers,sectionbib]{natbib}
\usepackage{wasysym}
\usepackage{textcmds}

\usepackage[all]{xy}

\spnewtheorem{task}[theorem]{Task}{\bfseries}{\itshape}
\spnewtheorem{heuristic}[theorem]{Heuristic}{\bfseries}{\itshape}
\spnewtheorem{example2}[theorem]{Example}{\bfseries}{}

\usepackage{amssymb}
\usepackage{tikz}
\usetikzlibrary{matrix}

\SetEndCharOfAlgoLine{}
\DecMargin{2.2mm}
\SetKwInput{Input}{Input}
\SetKwInput{Output}{Output}
\SetKwFor{While}{While}{do}{}
\SetKwFor{ForEach}{For each}{do}{}
\SetKwIF{If}{ElseIf}{Else}{If}{then}{else if}{else}{endif}
\SetKw{Return}{Return}
\SetKw{KwRet}{return}
\SetKw{Recurse}{Recurse}
\SetKwFor{For}{For}{do}{}
\SetKw{Break}{break}

\DeclareMathOperator{\End}{End}
\newcommand{\poly}{\mathrm{poly}}
\newcommand\GL{{\mathrm{GL}}}
\newcommand\Id{{\mathrm{Id}}}
\newcommand\Mat{{\mathrm{Mat}}}
\newcommand\Aut{{\mathrm{Aut}}}
\def\Z{{\mathbb Z}}
\def\F{{\mathbb F}}

\newcommand{\QA}{\mathcal{B}_{p,\infty}}
\newcommand{\ZZ}{\mathbb{Z}}
\newcommand{\QQ}{\mathbb{Q}}

\newcommand{\legendre}[2]{\genfrac{(}{)}{}{}{#1}{#2}}

\newcommand{\order}{\mathcal{O}}
\renewcommand{\O}{\order}

\def\issharing{1}  % set the variable to 0 if you want to have comments and to 1 to make them disappear

\usepackage{ifthen}
\newcommand{\sharing}[2]{%
\ifthenelse{\equal{\issharing}{1}}%
{{#1}}%
{{#2}}%
}

\ifthenelse{\equal{\issharing}{1}}
{
\newcommand{\PK}[1]{}
\newcommand{\GI}[1]{}
\newcommand{\IM}[1]{}
\newcommand{\CP}[1]{}
\newcommand{\AL}[1]{}
\newcommand{\MC}[1]{}
\newcommand{\ALcor}[1]{}

}{
\newcommand{\PK}[1]{\textcolor{olive}{{\sf (Peter's comment:} {\sl{#1})}}}

\newcommand{\AL}[1]{\textcolor{purple}{{\sf (Antonin's comment:} {\sl{#1})}}}
\newcommand{\ALcor}[1]{\textcolor{purple}{{\sf (Antonin's correction:} {\sl{#1})}}}
\newcommand{\MC}[1]{\textcolor{orange}{{\sf (Mingjie's comment:} {\sl{#1})}}}
\newcommand{\CP}[1]{\textcolor{green}{{\sf (Christophe's comment:} {\sl{#1})}}}

}

\title{Hidden Stabilizers, the Isogeny To Endomorphism Ring Problem and the Cryptanalysis of pSIDH}
\author{Mingjie Chen\inst{1} \and 
Muhammad Imran \inst{2} \and
Gábor Ivanyos \inst{3} \and
Péter Kutas \inst{1,6} \and
Antonin Leroux\inst{4,5} \and
Christophe Petit\inst{1,7}}
\authorrunning{M. Chen, M. Imran, G.Ivanyos, P.Kutas, A.Leroux, C. Petit}

\institute{University of Birmingham, UK \and Budapest University of Technology and Economics, Hungary  \and Institute   for Computer Science and Control, Eötvös Loránd Research Network, Hungary \and DGA-MI, Bruz, France \and IRMAR, UMR 6625, Université de Rennes, France \and Eötvös Loránd University, Hungary \and Université libre de Bruxelles, Belgium}

\begin{document}

\maketitle

\begin{abstract}
The \emph{Isogeny to Endomorphism Ring Problem} (IsERP) asks to compute the endomorphism ring of the codomain of an isogeny between supersingular curves in characteristic $p$ given only a \emph{representation} for this isogeny, i.e. some data and an algorithm to evaluate this isogeny on any torsion point. This problem plays a central role in isogeny-based cryptography; it underlies the security of 
pSIDH protocol (ASIACRYPT 2022) and it is at the heart of the recent attacks that broke the SIDH key exchange. Prior to this work, no efficient algorithm was known to solve IsERP for a generic isogeny degree, the hardest case seemingly when the degree is prime. 

In this paper, we introduce a new quantum polynomial-time algorithm to solve IsERP for isogenies whose degrees are odd and have $O(\log\log p)$ many prime factors. As main technical tools, our algorithm uses a quantum algorithm for computing hidden Borel subgroups, a group action on supersingular isogenies from EUROCRYPT 2021, various algorithms for the Deuring correspondence and a new algorithm to lift arbitrary quaternion order elements modulo an odd integer $N$ with $O(\log\log p)$ many prime factors to powersmooth elements.

As a main consequence for cryptography, we obtain a quantum polynomial-time key recovery attack on pSIDH. The technical tools we use may also be of independent interest.

\end{abstract}

\section{Introduction}

The problem of computing an isogeny between two supersingular elliptic curves is believed to be hard, even for a quantum computer. The assumption that this statement is true led to the idea of using isogenies to build post-quantum cryptography.

However, building actual cryptography from this principle is not easy and the security of concrete isogeny-based protocols is based on weaker versions of the isogeny problem, where the attacker is given more information. The nature of this additional information differs from one proposal to another but the heart of the problem remains the same. 

At the core of the cryptanalytic efforts to attack the isogeny problems lies another problem: the endomorphism ring problem, which requires to compute the endomorphism ring of a curve given in input. In fact, computing isogenies and computing endomorphism rings are computationally equivalent problem for supersingular curves~\cite{eisentraeger18,Wesolowski2021equivalencce}. 
However, this equivalence result %uses various tricks to navigate inside the graph of supersingular curves, and it 
does not fully answer  the following question : given a ``reasonable'' representation of an isogeny $\phi : E\rightarrow E'$ and the knowledge of the endomorphism ring of the starting curve $E$, can we always efficiently compute  the endomorphism ring of the codomain $E'$? This question leads to the following problem, where the exact definition of \textit{weak isogeny representation} will be given in Section~\ref{sec: isogeny rep}. 
%
%More generally, is the following problem %problem defined below as Problem~\ref{prob: isogeny representation to endo ring.} (that we call the IsERP) 
%easy? 

\begin{problem}[Isogeny to Endomorphism Ring Problem (IsERP)]\label{prob: isogeny representation to endo ring.}
    Let $E$ be a supersingular elliptic curve over $\mathbb{F}_{p^2}$ and let $\phi: E\rightarrow E_1$ be an isogeny of degree $N$ for some integer $N$.
    Given $ \End(E)$ and a weak {isogeny representation} for $\phi$, compute $\End(E_1)$. 
    % Suppose you can evaluate $\phi$ on any point $P$ of order coprime to $N$. The cost of the evaluation is the size of the representation of $\phi(P)$ (i.e., you can essentially evaluate $\phi$ either on points that are defined over small field extensions or points of powersmooth degree). Compute the endomorphism ring of $E_1$
\end{problem}

The answer to this question is known to be yes when the degree of $\phi$ is powersmooth (and this is what is used in the equivalence results mentioned above), but the question remains open for an arbitrary degree.
For a prime degree, this problem can be seen as the generalization of the key recovery problem for the pSIDH scheme recently introduced by Leroux \cite{leroux2021new}. 
The best known algorithm has subexponential quantum complexity in $N$, and the generic endomorphism ring attack has complexity exponential in $\log p$.  

%Our goal in this work, is to study this problem, highlight its link with other known problems in isogeny-based cryptography, and explore new directions to attack it. 

%Below, we give a historical overview of isogeny-based cryptography and illustrate how the various problems treated in this work are connected with one another. 

\subsubsection{Isogeny-based cryptography.}
Isogeny-based cryptography originates in Couveignes' seminal work \cite{couveignes1999hard} where he proposed to use the natural class group action on ordinary elliptic curves to instantiate a potentially quantum-resistant version of the Diffie-Hellman key exchange. The reasoning for that is that the discrete logarithm problem has more structure than needed to instantiate a key exchange, and this structure is exploited in Shor's algorithm \cite{Shor97}. Couveignes' ideas were rediscovered by Rostovtsev and Stolbunov \cite{rostovtsev2006public} and thus the resulting scheme is referred to as the CRS key exchange. These ideas were far from practical and a major breakthrough came with the invention of CSIDH \cite{castryck2018csidh}. The idea is quite similar but one uses supersingular elliptic curves defined over $\mathbb{F}_p$ and the acting group is the class group of $\mathbb{Z}[\sqrt{-p}]$. In other words one considers supersingular curves defined over $\mathbb{F}_p$ together with isogenies defined over $\mathbb{F}_p$ as well. 

%\subsubsection{SIDH and torsion point attacks.}
The same idea does not apply to supersingular curves defined over $\mathbb{F}_{p^2}$ because the endomorphism rings are non-commutative (and so the action of isogenies is not well-defined). This means that providing codomains of secret isogenies (i.e., curves $E,E_A,E_B$) is not enough to arrive at a shared secret that both parties can compute. Thus in order to instantiate a Diffie-Hellman-like key exchange on the full set of supersingular curves parties must provide additional information. In 2011 De Feo and Jao proposed SIDH \cite{jao2011towards} where both parties share the images of other person's torsion basis under their secret isogeny. 
This motivated the following problem: 
\begin{problem}
Let $E$ be a supersingular elliptic curve and let $A,B$ be coprime smooth numbers. Let $\phi:E\rightarrow E_A$ be a secret isogeny of degree $A$. One is provided with the action of $\phi$ on $E[B]$. Compute $\phi$.
\end{problem}

In \cite{petit2017faster} it was shown that this problem can be solved in polynomial time for certain parameter sets (where $B>p^2A^2$). 
In order to instantiate SIDH efficiently one usually uses parameters $A,B,p$ such that $AB$ divides $p+1$ as then all computations can be carried out over $\mathbb{F}_p^2$ so in some sense these initial results seemed theoretical. Then the initial idea of Petit \cite{petit2017faster} was improved in \cite{quehen2021improved} to $B>\sqrt{p}A^2$ which already included parameter sets which could have been used in SIDH variants. Nevertheless none of these attacks directly impacted SIDH where $A$ and $B$ are roughly the same size. Then in 2022 Castryck and Decru \cite{castryck2022efficient} (and independently Maino and Martindale \cite{maino2022attack}) vastly improved these using ingenious techniques (utilizing superspecial abelian surfaces) which break SIDH with known endomorphism ring in polynomial time even if $A$ and $B$ are balanced. Finally, Robert proposed a polynomial-time attack on SIDH with unknown endomorphism ring (furthermore, he only needs $B^2>A$ as opposed to $B>A$ in other attacks). 

These attacks have shown that using smooth degree isogenies and providing torsion point information will potentially not lead to secure and efficient cryptographic constructions (in \cite{fouotsa2023m} some countermeasures are proposed, but the ones that are not broken are much less efficient than the original SIDH construction). Thus, in order to navigate in the supersingular isogeny graph parties have to share some other kind of extra information.

\subsubsection{Alternative isogeny representations.}
In the pSIDH protocol introduced by Leroux~\cite{leroux2021new}, one reveals \emph{suborder representations} for isogenies of large prime degrees to build a key exchange. Suborder representations are a particular kind of weak isogeny representations, i.e. some data to represent isogenies together with an algorithm to efficiently evaluate these isogenies on any point up to a scalar. 
Prime degree isogenies were not really used before as one cannot write down the isogeny itself (but one can compute its codomain with non-trivial techniques). %But one can find a compact isogeny representation for it nonetheless, and Leroux showed how this isogeny representation was enough to perform an SIDH-like key exchange.  
More recently, a similar type of secret isogeny was used in the SCALLOP scheme \cite{de2023scallop}. In SCALLOP, a partial isogeny representation is revealed to the attacker. \CP{make more precise?}

From a cryptanalytic point of view, the unlimited amount of torsion information provided by the isogeny representation revealed in pSIDH (and more generally, any isogeny representation) is very interesting. However, when the kernel points are not defined over a small extension, the known algorithms do not apply and it is still unclear how to exploit the isogeny representation to recover the secret isogeny. 

Leroux studied the case where a specific isogeny representation (the suborder representation) is revealed, but we can generalize this setting to any isogeny representation. 
He showed that computing the endomorphism ring of the codomain would make pSIDH insecure, therefore motivating Problem~\ref{prob: isogeny representation to endo ring.} in the prime case. 

%He showed that the security of pSIDH comes down to computing the endomorphism ring of the codomain, and this is why the IsERP (that we introduced above as Problem~\ref{prob: isogeny representation to endo ring.}) is related to the security of pSIDH. 

% Previous attacks do not apply as they require the image points to have smooth degrees which have to be defined over small extension fields (none of which is the case in pSIDH). This motivates the following problem:

% \begin{problem}
%     Let $E$ be a supersingular elliptic curve over $\mathbb{F}_{p^2}$ and let $\phi: E\rightarrow E_1$ be an isogeny of degree $f$ where $f$ is a large prime number. Suppose you can evaluate $\phi$ on any point $P$ of order coprime to $f$. The cost of the evaluation is the size of the representation of $\phi(P)$ (i.e., you can essentially evaluate $\phi$ either on points that are defined over small field extensions). Compute the endomorphism ring of $E_1$
% \end{problem}

More recently, Robert introduced yet another isogeny representation based on torsion point images and the recent SIDH attacks~\cite{robert2022evaluating}. 
This representation could be used (for isogenies with large prime degrees) instead of the suborder representation to derive a key exchange protocol similar to pSIDH, and this protocol would be similarly affected by our new results. %\CP{Péter/Antonin to check this paragraph}

\subsubsection{A group action for SIDH and pSIDH}
In~\cite{kutas2021one} the authors introduce a group action on a particular set of supersingular elliptic curves. Let $E$ be a supersingular elliptic curve with endomorphism ring isomorphic to $O$. Then $(O/NO)^*$ acts naturally on the set of cyclic subgroups of $E$ of order $N$. If there is a one-to-one correspondence between cyclic subgroups and $N$-isogenous curves, then one can look at this action as acting on a set of curves. This action was used to provide a subexponential quantum key recovery attack on overstretched SIDH parameters. 

The reason the attack only works for overstretched parameter sets is that in general this group action is not easy to evaluate (thus substantial amount of extra information on the secret isogeny is needed). This motivates the following problem where the name Malleability Oracle Problem comes from the term introduced in~\cite{kutas2021one}.

%Cyclic subgroups correspond to $N$-isogenous elliptic curves thus if there is a one-to-one correspondence between cyclic subgroups and $N$-isogenous curves, then one can view this group acting on the set of $N$-isogenous elliptic curves by defining $\sigma* A= E/\sigma(A)$ where $\sigma\in O$ and $A\in E[N]$. In the context of isogeny-based cryptography the interesting problem is when one is given a particular codomain $E_A\cong E/A$ of an $N$-isogeny but the corresponding isogeny is not provided and one has to evaluate this group action. We formulate this in the following fashion:
%In isogeny-based cryptography one is usually interested in the following problem. Let $E_1,E_2$ be supersingular elliptic curves defined over $\mathbb{F}_{p^2}$ and suppose there exists a degree $d$ isogeny in between. The pure isogeny problem asks to find a degree $d$ isogeny between $E_1$ and $E_2$. In a cryptographic context $d$ has to be large and smooth for both the problem to be hard and to be able to represent the isogeny efficiently. 

%Suppose that the endomorphism ring of $E_1$ is known. Note that if the endomorphism ring of $E_2$ is also known then one can compute an isogeny between the two curves \cite{kohel2014quaternion}, although not necessarily of the correct degree only if additional information is known (the isogeny is short \cite{galbraith_security_2016}, or one is provided some torsion point information \cite{fouotsa2022isogeny}). In \cite{kutas2021one} the authors consider the following problem: 

\begin{problem}[Malleability Oracle Problem]\label{Prob:endev}
    Let $E$ be a supersingular elliptic curve and let $\phi:E\rightarrow E'$ be a secret isogeny with kernel generated by $A$. Let $\sigma\in  \End(E)$. Find the $j$-invariant of $E/\langle \sigma(A) \rangle$. 
\end{problem}

%Why is Problem \ref{prob:1} interesting from a cryptanalysis point of view? In \cite{kutas2021one} it was shown that (under certain conditions) if one can evaluate this group action, then there is a subexponential quantum algorithm for retrieving the secret isogeny $\phi$. Problem \ref{prob:1} was originally studied in the context of SIDH \cite{kutas2021one} and torsion-point information was used to evaluate the group action. As mentioned before, new countermeasures might occur thus it is important to have Problem \ref{prob:1} as a cryptanalysis tool and it is also important to understand how it relates to more standard hard problems in isogeny-based cryptography (e.g., computing endomorphism rings).
%The authors show that there is a subexponential quantum algorithm that solves Problem \ref{prob:1}. The key idea is to reduce this problem to an abelian hidden shift problem for which one can apply Kuperberg's algorithm \cite{kuperberg2005subexponential}. The group action comes from the following observation. One has that $End(E)/N\cdot End(E)$ is isomorphic to $M_2(\mathbb{Z}_N)$ and thus $(End(E)/N\cdot End(E))^*$ is isomorphic to $GL_2(\mathbb{Z}_N)$. This will not be an abelian group so instead one takes a subgroup (more precisely a subgroup of $PGL_2(\mathbb{Z}_N)$) of this group to make the reduction work. This requires a lot of tedious work and one throws away a significant amount of information.

\CP{motivate and define quaternion problem here}

\subsubsection{Contributions.}
Our main result is the following theorem on the resolution of IsERP.  %\CP{check consistency with claims elsewhere}. 

\begin{theorem}\label{thm:solvingIsERP}
   Let $N = \prod \ell_i^{e_i} \neq p$ be an odd integer that is of size polynomial in $p$ and has $O(\log(\log p))$ divisors. Then there exists a quantum polynomial-time algorithm that solves the IsERP.
\end{theorem}

We first provide a reduction from the IsERP to the Powersmooth Quaternion Lifting Problem (PQLP). The PQLP is the problem of finding a powersmooth representative for a given class in $\O / N\O$ for some integer $N$ and maximal order $\O$ in the quaternion algebra $\QA$. 
%Given that the PQLP is a pure quaternion problem, it seems like a good candidate to try to attack the IsERP. We illustrate this, by using our reduction to obtain a quantum polynomial-time algorithm requiring a subexponential precomputation (only depending on the starting curve and the degree) to solve the IsERP.
%While this does not improve the best known algorithm when only one instance of the IsERP is given, it shows that one should not reuse the same starting curve and degree $N$.  

Our reduction from the IsERP to PQLP is obtained through a quantum equivalence between the IsERP and a problem similar to Problem~\ref{Prob:endev}, which we call the Group Action Evaluation Problem. The most difficult direction of this equivalence (reducing the IsERP to the Group Action Evaluation Problem) is obtained with a quantum polynomial-time algorithm. The other reduction is classical and uses  standard tools for the Deuring correspondence.  

The quantum polynomial reduction relies on a special case of the well-known hidden subgroup problem (HSP), namely when the acting group is $\GL_2(\Z/N\Z)$ and the hidden subgroup is a conjugate of the subgroup of upper triangular matrices. 
This problem was previously studied only for prime $N$ \cite{DMS10}
and in this paper we provide a polynomial-time quantum algorithm for any $N$. Furthermore, whenever $N$ is smooth we propose a classical polynomial-time algorithm which might be of independent interest.  

We then propose a classical polynomial-time algorithm for the PQLP. The algorithm relies on several tools developed in KLPT \cite{kohel2014quaternion}. Namely we decompose elements $\sigma\in \mathcal{O}$ as $\alpha_1\gamma\alpha_2\gamma\alpha_3$ where $\alpha_i$ lie in a special subset of $\mathcal{O}$ (linear combinations of $j,ij$) that can be lifted efficiently to powersmooth elements, and $\gamma$ is a fixed element of $\mathcal{O}$ of powersmooth norm. 
Finding $\gamma$ and lifting the $\alpha_i$ are accomplished with slightly modified subroutines of KLPT, whereas the decomposition itself is inspired by similar decompositions in other contexts~\cite{petit2008full}.
The lifting algorithm requires that $N$ is odd and has $O(\log\log p)$ prime factors. We look at approaches to generalize this algorithm to arbitrary $N$ (thus solve IsERP for arbitrary degrees) in Appendix \ref{app:hybridPQLP}. We have also implemented this algorithm for prime $N$ in Sagemath~\cite{sagemath}, available on GitHub~\cite{github}.

\medskip
The rest of the paper is organized as follows: 
in Section~\ref{sec: prelim}, we introduce some necessary background. Then, in Section~\ref{sec: borel hsp}, we introduce a quantum algorithm to solve the Borel Hidden Subgroup Problem. In Section~\ref{sec: isogeny section} we define the Group Action Evaluation Problem and the Powersmooth Quaternion Lift Problem (PQLP). We show various reductions between the two problems and the IsERP, most importantly reducing IsERP to the PQLP.  In Section \ref{sec:respqlp} we describe our polynomial-time algorithm for PQLP, which leads to a resolution of the IsERP through the reductions. Finally in Section \ref{sec:conclusion}, we discuss the impacts of our results on isogeny-based cryptography.

\section{Preliminaries}
\label{sec: prelim}

Below, we give a brief introduction to some necessary mathematical background. More details on elliptic curves and isogenies can be found in \cite{silverman2009arithmetic}. The book of John Voight \cite{voight2018quaternion} is a good reference regarding quaternion algebras and the Deuring correspondence.   
In the remaining of this paper, we fix a prime $p > 2$.

\subsection{Supersingular elliptic curves and isogenies}
Let $E_1,E_2$ be elliptic curves defined over a finite field $\mathbb{F}_q$. An isogeny is a non-constant rational map from $E_1$ to $E_2$ that is simultaneously a group homomorphism. Equivalently, it is a non-constant rational map that sends the point of infinity of $E_1$ to the point of infinity of $E_2$. An isogeny induces a field extension $K(E_1)/K(E_2)$ of function fields. An isogeny is called separable, inseparable or purely inseparable if the extension of function field is of the respective type. The degree of the isogeny is the degree of the field extension $K(E_1)/K(E_2)$. The kernel of an isogeny $\phi:E_1\rightarrow E_2 $ is a finite subgroup of $E_1$. If the isogeny is separable, then the size of the kernel is equal to the degree of the isogeny (more generally, the size of the kernel equals the separable degree of the field extension $K(E_1)/K(E_2)$). 
For every isogeny $\phi: E_1\rightarrow E_2$ there exists a dual isogeny $\hat{\phi}:E_2\rightarrow E_1$ such that $\deg(\phi)=\deg(\hat{\phi})=d$ and $\phi\circ\hat{\phi}=[d]_{E_2}$ (and $\hat{\phi}\circ\phi=[d]_{E_1}$). Isogenies (together with the zero map) from $E$ to itself are called endomorphisms. 
Endomorphisms of an elliptic curve form a ring under addition and composition. An elliptic curve over a finite field is called ordinary if its endomorphism ring is commutative, and supersingular otherwise.

\subsection{Quaternion algebras}

The endomorphism rings of supersingular elliptic curves over $\F_{p^2}$ are isomorphic to maximal orders of $B_{p,\infty}$, the quaternion algebra ramified at $p$ and $\infty$. We fix a basis $1,i,j,k$ of $B_{p,\infty}$, satisfying $i^2=-q$, $j^2=-p$ and $k=ij=-ji$ for some integer $q$. The canonical involution of conjugation sends an element $ \alpha = a +ib + jc + kd $ to $\overline{\alpha}  = a - (ib + jc +kd)$. 
A \textit{fractional ideal} $I$ in $B_{p,\infty}$ is a $\Z$-lattice of rank four. We denote by $n(I)$ the \emph{norm} of $I$ as the largest rational number such that $n(\alpha) \in n(I)\Z$ for any $\alpha \in I$. 
An order $\O$ is a subring of $B_{p,\infty}$ that is also a fractional ideal. An order is called \textit{maximal} when it is not contained in any other larger order. 
The left order of a fractional ideal is defined as $\O_L(I) = \{\alpha \in B_{p,\infty}\;|\; \alpha I \subset I \}$ and similarly for the right order $\O_R(I)$. Then $I$ is said to be a right $\O_R(I)$-ideal or a left $\O_L(I)$-ideal. A fractional ideal is \textit{integral} if it is contained in its left order, or equivalently in its right order; we refer to integral ideals hereafter as ideals. 
% Two left $\O$-ideals $I$ and $J$ are equivalent if there exists $\beta \in B_{p,\infty}^\times$, such that $I = J \beta$. For a given $\O$, this defines equivalences classes of left $\O$-ideals, and we denote the set of such classes by Cl$(\O)$. We will reuse the following notation from \cite{FKLPW20}: for any ideal $K$ and any $\alpha \in \QA^\times$, we write $\chi_K (\alpha) = K \overline{\alpha}/n(K)$. Ideals equivalent to $K$ are precisely the ideals $\chi_K (\alpha)$ with $\alpha \in K\setminus\{0\}$.
Eichler orders are the intersection of two maximal orders. If $I$ is an ideal, we can define the Eichler order associated to $I$ as $\O_L(I) \cap \O_R(I)$. In that case, it can be shown that $\O_L(I) \cap \O_R(I) = \Z + I$ (see \cite{sqisign}).

\subsection{The Deuring correspondence}
Fix a supersingular elliptic curve $E_0$, and an order $\O_0 \simeq \End(E_0)$.
The curve/order correspondence allows one to associate to each outgoing isogeny $\varphi : E_0 \rightarrow E_1$ an integral left $\O_0$-ideal, and every such ideal arises in this way (see \cite{kohel1996endomorphism} for instance). Through this correspondence, the ring $\End(E_1)$ is isomorphic to the right order of this ideal. This isogeny/ideal correspondence is defined in \cite{waterhouse1969abelian}, and in the separable case, it is explicitly given as follows.   

\begin{definition}
\label{corresponding ideal and isogeny}
Given $I$ an integral left $\O_0$-ideal coprime to $p$, we define the $I$-torsion $E_0[I]= \lbrace P \in E_0(\overline{\mathbb{F}}_{p^2}): \alpha(P)= 0 \mbox{ for all } \alpha \in I \rbrace$. To $I$, we associate the separable isogeny $\varphi_I$ of kernel $E_0[I]$. 
Conversely given a separable isogeny $\varphi$, the corresponding ideal is defined as $I_\varphi = \lbrace \alpha \in \O_0 \;:\; \alpha(P) = 0 \mbox{ for all } P \in \textnormal{ker}(\varphi) \rbrace$.
\end{definition}

We summarize properties of the Deuring correspondence in Table~\ref{tab: deuring correspondence}, borrowed from \cite{sqisign}. 

\begin{table}[]
\centering
\begin{tabular}{ l @{\hspace{2em}} l }
\hline
 Supersingular $j$-invariants over $\mathbb{F}_{p^2}$ & Maximal orders in $\QA$ \\
 $j(E)$ (up to Galois conjugacy) & $\order \cong \textnormal{End}(E)$ (up to isomorphism)\\
 \hline
 $(E_1,\varphi)$ with $\varphi: E \rightarrow E_1$ & $I_\varphi$ integral left $\order$-ideal and right $\O_1$-ideal\\
 \hline
 $\theta \in \End(E_0)$ & Principal ideal $\O \theta$ \\
 \hline
 deg$(\varphi)$ & $n(I_\varphi)$ \\
 \hline
 % $\hat{\varphi}$ & $\overline{I_\varphi}$ \\
 % \hline
 % $\varphi : E\rightarrow E_1,\psi : E \rightarrow E_1$ & Equivalent Ideals $I_\varphi \sim I_\psi$ \\
 % \hline
 % Supersingular $j$-invariants over $\mathbb{F}_{p^2}$ & $\Cl(\O)$ \\
 % \hline
 % $\tau \circ \rho : E \rightarrow E_1 \rightarrow E_2 $ & $I _{\tau \circ \rho} = I_\rho \cdot I_\tau $ \\
 % \hline
% \cite{} & Eichler orders $\frakO = \O \cap \O_1$ of level $N$ \\
% \hline
 % $N$-isogenies (up to isomorphism) & $\Cl(\frakO)$, with Eichler order $\frakO$ of level $N$ \\
% $[$this work, \cref{prop: eichler class group interpretation} $]$ & \\
% [ & hello \\
 % \hline
\end{tabular}
\caption{The Deuring correspondence, a summary~\cite{sqisign}. \label{tab: deuring correspondence}}
\end{table}

\subsection{Isogeny representation}
\label{sec: isogeny rep} 

In this subsection, we look at isogenies through a more algorithmic prism. Specifically, we consider the following question: what does it mean to ``compute'' an isogeny?
A natural answer is a rational map representation of the isogeny. 
Other representations are however possible, and  in~\cite[Sec~2.4]{Petit2017Hard} and~\cite{leroux2021new} it is argued that any such representation should allow efficient evaluation at arbitrary points (for a more complete study, look at \cite[chapter 4]{leroux2022quaternion}).  
More formally, Leroux defines an isogeny representation as some data $s_\phi$ associated to an isogeny $\phi :  E\rightarrow E'$ of degree $N$ such that there are two algorithms: one to \qq{verify} and one to \qq{evaluate} $\phi$.

The motivation to have a verification algorithm is found in a cryptographic context where an adversary might try to cheat by revealing something that is not a valid isogeny representation. But, in the more cryptanalatic point of view of this paper, we can assume that we work with a valid isogeny representation. This is why we take a relaxed definition of isogeny representation where we only require an evaluation algorithm (a verification algorithm can probably be derived from the evaluation algorithm anyway). Moreover, we assume that the representation is ``efficient'' meaning that is has polynomial size and the evaluation algorithm is polynomial-time in the log of the degree and the prime. We give a detailed version below. In our context, it is sufficient that the evaluation algorithm gives evaluation of points up to a (common) scalar which is why we qualify our isogeny representation as \textit{weak}. 

\begin{definition}
\label{def: isogeny rep}
A \textit{weak isogeny representation} for the isogeny $\phi : E \rightarrow E'$ of degree $N$, is a data $s_\phi$ of size $O(polylog(p+N))$ (associated to a unique isogeny $\phi$), such that there exists an algorithm $\mathcal{E}$ that takes $s_\phi$ and a point $P$ of the curve $E$ of order $d$ in input and computes $\lambda(d) \phi(P)$ in $O( polylog(d+ N + |P|))$ for any point $P$ of $E$, where $|P|$ is the bitsize of the representation of $P$.   
\end{definition}

The notion of isogeny representation is particularly relevant when the degree $N$ is a big prime and the kernel points are defined over an $\F_p$-extension of big degree (this is exactly the setting of pSIDH \cite{leroux2021new}). Indeed, in that case, the standard ways to represent isogenies (with polynomials, or kernel points) are not compact or efficient enough to match our definition.

The Deuring correspondence gave us the tools to obtain efficient representations with a natural isogeny representation obtained by taking $s_\phi$ as the ideal $I_\phi$ corresponding to $\phi$. This ideal representation matches Definition~\ref{def: isogeny rep}, however it also reveals the endomorphism ring of $E'$. 
One of the motivations of Leroux in \cite{leroux2021new} to introduce another isogeny representation (called the suborder representation) is to have an isogeny representation that does not directly reveal the endomorphism ring of the codomain. This suborder representation matches our notion of weak isogeny representation as defined in Definition~\ref{def: isogeny rep}. 
The main contribution of this paper implies that the suborder representation does not hide the endomorphism ring of the codomain to a quantum computer, even when the degree is prime.  

Since then, Robert \cite{robert2022evaluating} suggested to use the techniques introduced to attack SIDH in order to obtain another isogeny representation (this one not even requiring to reveal the endomorphism ring of the domain). 
Our analysis holds for any suborder representation, hence it also applies to Robert's one.
%While this idea is nice, our analysis is made without any assumption regarding the concrete isogeny representation used, and this is why we do not particularly benefit from this result.      

\subsection{The pSIDH key exchange}
\label{sec: psidh}

As an application of the hardness of computing the endomorphism ring from the suborder representation, Leroux introduced a key exchange called pSIDH. 
The principle can be summarized as follows: use the evaluation algorithm for the suborder representation to perform an SIDH-like key exchange, but for isogenies of big prime degree. The SIDH and pSIDH key exchange both use the following commutative isogeny diagram:   

\[
\xymatrix{
E_B \ar[r]^{\psi_A} & E_{AB} \\
E_0\ar[u]^{\phi_B} \ar[r]^{\phi_A} & E_A \ar[u]_{\psi_B}
}
\]

In pSIDH, Alice and Bob's secret keys are ideal representations for the isogenies $\phi_A$ and $\phi_B$ (or equivalently the endomorphism ring of the two curves $E_A$ and $E_B$), and their associated public keys are the suborder representations for $\phi_A$ and $\phi_B$.

Leroux showed that the knowledge of $\End(E_A)$ (resp. $\End(E_B)$) and the suborder representation of $\phi_B$ (resp. $\phi_A$) was enough to compute the end curve $E_{AB}$ from which the common secret can be derived efficiently. 
The mechanism behind this computation is quite complicated and is not relevant for us since we target the key recovery problem. We refer to \cite{leroux2021new} for more details.

\subsection{The hidden subgroup problem}\label{subsec:HSP}
The hidden subgroup problem (HSP for short) in a group $G$ is defined as the problem of finding a subgroup $H\leq G$ given a function $f$ on $G$ satisfying that $f$ is constant on the left
cosets of $H$ and takes different values on different cosets, i.e., $f(x)=f(y)$ if and only if 
$x^{-1}y\in H$. There is also a right version of 
the hidden subgroup problem where the level sets of the hiding function $f$ are the right cosets of $H$. As taking inverses in $G$ maps left cosets to right cosets and vice versa, the two versions of HSP are equivalent. (One just needs to replace the hiding function with its composition with taking inverses.) Although the equivalence is straightforward, it is useful as in certain cases it is easier to understand
right cosets than left ones (or conversely).

The framework of HSP captures many computational problems including some problems which most cryptographic protocols used today rely on, e.g., factoring and the discrete logarithm problem. 
Shor's quantum algorithms \cite{Shor97} can solve factoring and the discrete logarithm problem efficiently. Furthermore, quantum polynomial time algorithms for the finite abelian HSP generalizing Shor's algorithm are available, see \cite{kitaev1995quantum},\cite{boneh1995quantum}.

It is well known that the graph isomorphism problem and the shortest vector problem can be cast as HSP in symmetric groups and dihedral groups, respectively. However, in contrast to the abelian case, there are only a few positive results known for HSP in finite non-commutative groups. As shown in \cite{ettinger2000quantum}, HSP in dihedral groups 
is related to another problem called the hidden shift problem. The hidden shift problem in a group $G$ is the problem of finding an element $s\in G$ given two functions $f_1$ and $f_2$ on $G$ satisfying that $f_1(g)=f_2(gs)$ for every $g\in G$. 
If $f_1$ and $f_2$ are injective then the hidden shift problem in an abelian group $G$ is
equivalent to a hidden subgroup problem in the semidirect product $G\rtimes \Z/2\Z$. This is of particular interest in isogeny contexts, as the key recovery problem in CSIDH can be reduced to the injective hidden shift problem in abelian groups in order to produce quantum subexponential-time attacks based on Kuperberg's algorithm \cite{kuperberg2005subexponential}.

In this paper, we consider a restricted HSP in the general linear group. We use the term  Borel hidden subgroup problem for it. \CP{explain name?}
%This problem appears to be of special interest as it will be shown that some variants of isogeny problem can be tackled using this framework. Therefore, this gives an alternative approach to attack isogeny-based cryptography.   
\begin{problem}
Let $N\in \Z_{\geq 1}$ and let $\Z/N\Z$ be the group of integers modulo $N$. The Borel HSP is the hidden subgroup problem in the general linear group $\GL_n(\Z/N\Z)$ for $N\in \Z_{\geq 1}$, i.e., the group of invertible $n$ by $n$ matrices with entries from $\Z/N\Z$, where the hidden subgroup $H$ is promised to be a conjugate of the subgroup consisting of the upper triangular matrices.
\end{problem}

Restricting the possible hidden subgroups in non-abelian groups may lead to efficient algorithms to find them. Denney et al. in \cite{DMS10} proposed a polynomial-time algorithm for the Borel HSP in $\GL_2(\mathbb{F}_p)$ for prime numbers $p$. A quantum algorithm for the more general case of $\GL_n(\F_q)$ over fields of size $q=p^k$, is provided by Ivanyos in \cite{ivanyos2011finding}. That algorithm runs in polynomial time if $q$ is not much smaller than $n$. 

In this paper, we consider the Borel HSP for $\GL_2(\Z_N)$ for any integer $N$ greater than one, and we present both classical and quantum algorithms for different parameters $N$. Note that $\GL_2(\Z/N\Z)$ acts as a permutation group on the set of the free cyclic $\Z/N\Z$-submodules of $(\Z/N\Z)^2$ and each Borel subgroup $H$ in $\GL_2(\Z/N\Z)$ is the stabilizer of a free cyclic $\Z/N\Z$-submodule $S$ of $(\Z/N\Z)^2$, thus finding the Borel subgroup is equivalent to finding the corresponding cyclic submodule. 
The main tool of the classical algorithm is a testing procedure to determine whether elements 
of $(\Z/N\Z)^2$ are in $S$. The classical algorithm solves the Borel HSP efficiently for any smooth number $N$, while the quantum algorithm efficiently solves the Borel HSP for arbitrary $N$. 
The main idea of the quantum algorithm is based on the observation that the problem can be reduced to another restricted hidden subgroup problem in the group $G=(\Z/N\Z)^2 \rtimes (\Z/N\Z)^*$ where the hidden subgroup is promised to be a complement of the normal subgroup $(\Z/N\Z)^2$. The latter restricted HSP can be cast as an instance of the multiple shift problem considered in~\cite{IPS18}, which can itself be seen as a generalization of the hidden shift problem.

\begin{problem}
The hidden multiple shift problem HMS($N,n,r$) is parameterized 
by three positive integers $N, n$ and $r,$ where $N>1$ and 
$2\leq r\leq N-1$. Assume that we have a set $H\subseteq \Z/N\Z$ 
of cardinality $r$ 
and a function $f_s: (\Z/N\Z)^n \times H \to \{0,1\}^l$, 
defined as $f_s(x,h)=f(x-hs)$ where 
$s\in (\Z/N\Z)^n$ and $f:(\Z/N\Z)^n\to \{0,1\}^l$ is an injective function. 
Given $f_s$ by an oracle, the task is 
to find $s \bmod{\frac{N}{\delta(H,N)}}$, 
where $\delta(H,N)$ is defined as the largest 
divisor of $N$ such that $h-h'$ is divisible by $\delta(H,N)$ 
for every $h,h'\in H$. 
\end{problem}

A special case of the HMS problem was first considered by Childs and van Dam \cite{childs2005quantum}. They presented a quantum polynomial time algorithm for the case when $n=1$ and $H$ is a contiguous interval of size $N^{\Omega(1)}$. For general $n$, an algorithm in \cite{IPS18} solves HMS($N,n,r$) in $O(\poly(n)(\frac{N}{r})^{n+O(1)})$. For a set $H$ of small size, HMS is close to the hidden shift problem. 
Specifically, HMS($N,n,2$) is the standard hidden shift problem. On the other extreme, for $r=N$, HMS is an abelian hidden subgroup problem in the group $(\Z/N\Z)^{n+1}$. Intuitively, the larger $r$ is, the easier HMS($N,n,r$) becomes. Below we restate the above mentioned result from \cite{IPS18} for the special case $n=1$.

\begin{theorem}\label{hms}
There is a quantum algorithm that solves the HMS($N,1,r$) 
in time $\left(\frac{N}{r}\right)^{O(1)}$ with high probability.
\end{theorem}

\subsection{The malleability oracle}

In \cite{kutas2021one} the authors introduce a general framework dubbed the malleability oracle. Let $G$ be a group acting on a set $X$ and let $f:X\rightarrow I$ be an injective function where $I$ is some set. The input of the malleability oracle is an element $g\in G$ and a value $f(x)$ ($x$ is not provided) and the output is $f(g*x)$. It is shown in \cite[Theorem 3.3]{kutas2021one} that if $G$ is abelian and the action of $G$ on $X$ is free and transitive then inverting $f(x)$ can be reduced to an abelian hidden shift problem. The idea of the proof is as follows. One takes an arbitrary known $x_0$ and the corresponding $f(x_0)$. Then one can define two functions $f_0,f_1$ from $G$ to $I$ where $f_0(g)=f(g*x_0)$ and $f_1(g)=f(g*x)$. These functions are well-defined as $f$ was injective. Now since the action of $G$ is transitive there is an element $s$ that takes $x_0$ to $x$. One can easily see that $f$ and $f_0$ are shifts of each other and the shift is realized by that element $s$. Since the action is free, $f$ and $f_0$ will be injective functions themselves hence one can apply Kuperberg's algorithm to find $s$ and finally that is enough to compute $x$. 
\begin{remark}
It follows from the proof that it is not strictly necessary for the action to be transitive. It is enough if we know any element in the orbit of the secret $x$. For instance if suffices if there are only a few orbits and we have a representative of each of them (as we can run Kuperberg's algorithm multiple times with different $x_0$s). 
\end{remark}
This framework is pretty general and applies to a variety of scenarios. It encompasses the group action underlying CSIDH \cite{castryck2018csidh} but in \cite{kutas2021one} it is shown how this framework can be applied in the context of sidh variants.

Let $E$ be a supersingular elliptic curve with known endomorphism ring $O$. Here we assume that one can evaluate every element of $O$ efficiently on points of $E$. Let $N$ be any integer. Then $O/NO$ is isomorphic to $M_2(\mathbb{Z}/N\mathbb{Z})$ \cite[Theorem 42.1.9]{voight2018quaternion}. This implies that $(O/NO)^*$ is isomorphic to $G=\GL_2(\mathbb{Z}/N\mathbb{Z})$. 
Now it is clear that $G$ acts on cyclic subgroups of order $N$ of $E$ by evaluation. When there is a one-to-one correspondence between cyclic subgroups of order $N$ and $N$-isogenous curves to $E$, then this implies an action on $N$-isogenous curves to $E$. 
What would a malleability oracle look like in this framework? One is given a curve $E'$ that is $N$-isogenous to $E$. Let $A$ be the corresponding secret kernel. Now the input of the oracle is an endomorphism $\sigma$ (whose degree is coprime to $N$) and then it returns $E/\sigma(A)$. 
The above theorem \qq{almost} states that if one has access to such an oracle, then one can compute $A$ via a hidden shift algorithm. The \qq{almost} part comes from the fact that $G$ here is not abelian and the group action is not free. In \cite{kutas2021one} it is shown that one can get around this issue by essentially just utilizing a subgroup of $G$ that is abelian (and evoking some small technical conditions). 

One can look at this result as a subexponential quantum reduction from finding a certain $N$-isogeny to being able to instantiate the malleability oracle, which is formulated as Problem~\ref{Prob:endev}. The results of Section~\ref{sec: isogeny section} will be related to a generalization of \ref{Prob:endev}.

In \cite{kutas2021one} the authors were able to solve the above problem when $\deg(\phi)=2^k$ and the action of $\phi$ is known on a sufficiently large subgroup of $E$. In order to achieve this result one had to throw away most of the available information (by restricting $G$ to a small abelian subgroup) in order to fit the malleability oracle framework. In this paper we show that utilizing the entire $G$-action improves on \cite{kutas2021one} significantly.

\section{The Borel hidden subgroup problem}
\label{sec: borel hsp}
In this section, we present both classical and quantum algorithms for the \qq{two-dimensional} Borel hidden subgroup problem. The classical algorithm solves the Borel HSP efficiently in the group $\GL_2(\Z/N\Z)$ for smooth number $N$, while the quantum algorithm solves it efficiently for any positive odd number $N$. For an even number $N$ we can use a classical procedure applied to the 2-part of $N$ with the quantum one for the odd part of $N$ to obtain a quantum method for every $N$.

Let $N$ be an integer greater than one and let $R=\Z/N\Z$. Let $V$ denote $R^2$, the free $R$-module of rank $2$, let $E=\End_{R}(V)$, and $G=\Aut(V)=E^*$. By fixing a basis, we have an explicit isomorphism $E\cong \Mat_2(R)$ and $G\cong \GL_2(R)$. 
Note that $G$ acts as a permutation group on the set of the free cyclic $R$-submodules of $V$. Let $H$ be the stabilizer of a secret free cyclic submodule $S$. In the matrix notation, $H$ is a conjugate of the subgroup consisting of the upper triangular matrices in $G$. That is, in an appropriate basis for $V$, the elements of $H$ are of the form
$$\begin{pmatrix}
* & * \\
0 & *
\end{pmatrix},
$$ 
where the diagonal entries are units in $R$. (Here the first basis element is a generator for $S$.) 
The Borel HSP in $G$ is the following:
we are given a function on $G$ (given by an oracle) that is constant on the left cosets of $H$ and takes different values on distinct cosets, the task is to find $H$, or equivalently the submodule $S$. 

Using Chinese remaindering, one can reduce the case when $N$ is any number of known factorization to instances of the prime power case.
\begin{lemma}\label{crt}
Let $N=N_1N_2$ be a known decomposition of $N$ where
$\gcd(N_1,N_2)=1$. Then we have 
$$\GL_2(\Z/N\Z)\cong \GL_2(\Z/N_1\Z)\times \GL_2(\Z/N_2\Z).$$
Moreover, one can reduce the Borel
HSP in $\GL_2(\Z/N\Z)$ to the Borel HSP in 
$\GL_2(\Z/N_i\Z)$ for $i=1,2$.
\end{lemma}
\begin{proof}
By the Chinese Remainder Theorem, 
$\Z/N\Z \cong \Z/N_1\Z\oplus \Z/N_2\Z$,
$(\Z/N\Z)^2 \cong (\Z/N_1\Z)^2\oplus (\Z/N_2\Z)^2$, 
$\End((\Z/N\Z)^2) \cong \End((\Z/N_1\Z)^2)\oplus \End((\Z/N_2\Z)^2)$.
Furthermore, these isomorphisms can be efficiently computed using the extended Euclidean algorithm. The restriction of the third isomorphism also gives  $\Aut((\Z/N\Z)^2)\cong \Aut((\Z/N_1\Z)^2)\times \Aut((\Z/N_2\Z)^2)$. The stabilizer $H$ of the free cyclic submodule $S$ generated by $(A_1,A_2)\in (\Z/N_1\Z)^2\oplus (\Z/N_2\Z)^2$ is the direct product of the stabilizers $H_i$ of $S_i$, where $S_i$ are the free cyclic submodules over $\Z/N_i\Z$ generated by $A_i$. Hiding functions for $H_i$ can be obtained by restricting the hiding function for $H$ to the component $\Aut((\Z/N_i\Z)^2)$.
\qed 
\end{proof}

\subsection{A classical Borel HSP algorithm}

Based on iterated applications of Lemma~\ref{crt}, we can focus on the prime power case. (Note that the factorization of $N$ can be computed in deterministic time polynomial in $B\log N$ where $B$ is an upper bound on the prime divisors of $N$.) 
Therefore, we assume $N=q^k$ for a prime number $q$. 

\CP{I struggle to parse the description in this section (and I think one crypto reviewer had the same complain). I can see four different tests 1) whether $u\in S$ (described in a paragraph) 2) whether $\varphi(S)\in S$ (described in a lemma) 3) the so-called testing procedure (not sure what it tests) 4) within the description of the testing procedure, another test for the case $q=2$. IMHO, the structure does not make clear how these tests relate to each other. Suggestions: state explicitly what the testing procedure is supposed to do, add pseudocode for it, isolate the $q=2$ case like the other ones, and explain how the testing procedure relates to the other tests described }

An important subroutine in our algorithm is a procedure for testing whether an element $u\in V$ is in $S$ based on the following observations.
If $u\in S$ then for any $\varphi\in E$ such that $\varphi+\Id\in G$ and $\varphi(V)\leq Ru$ we have $\varphi+\Id\in H$.
This is because for $u\in S$ we have $\varphi(u)\in R u\leq S$ and $\Id(u)=u\in S$. 
On the other hand, if $u\not\in S$ then there exists an element $\varphi\in E$ with $\varphi(V)\leq Ru$ and $\varphi(S)\not\leq S$. Indeed, if $\{v,w\}$ is an $R$-basis of $V$ such that $v$ is a generator of $S$, then the map sending $v$ to $u$ and $w$ to zero satisfies these properties.

Another ingredient of the testing procedure is the following.
\begin{lemma}
If $\varphi+\Id\in G$ then $\varphi(S)\leq S$ if and only if
$\varphi+\Id\in H$.
\end{lemma}
\begin{proof}
If $\varphi(S)\leq S$ then $(\varphi+\Id)S \leq
\varphi(S)+S=S$. To see the reverse implication, assume that
$\varphi(v)\not\in S$ for some $v\in S$. Then $\varphi(v)+v$ is in the coset $\varphi(v)+S$ disjoint from $S$.
\qed 
\end{proof}
Thus for $\varphi$ with $\varphi+\Id\in G$ we can test whether
$\varphi(S)\leq S$ by comparing the value of the hiding function taken on $\varphi+\Id$ with that on $\Id$. 

\begin{paragraph}{Testing procedure:}
Let $w_1, w_2$ be a fixed basis of $V$. Given $u\in V$ we define two maps $\varphi_1, \varphi_2$ by $\varphi_i(w_i)=u$ and $\varphi_i(w_{3-i})=0$. Note that $\varphi_1$ and $\varphi_2$ generate $E_u:=\{\varphi\in E:\varphi(V)\leq Ru\}$ as an $R$-submodule of $E$. Therefore if $\varphi_i(S)\leq S$ ($i=1,2$) then for every element $\varphi\in E_u$ we have $\varphi(S)\leq S$. If $u\in qV$ then $\varphi_i-\Id\in qE-\Id\subseteq G$ ($i=1,2$), so we can test whether $u\in S$ by testing $\varphi_i(S)\leq S$ ($i=1,2$) by comparing the value of the hiding function taken on $\varphi_i+\Id$ with that on $\Id$. If $q\neq 2$ then either $\varphi_i-\Id$ or
$-\varphi_i-\Id$ (or both) fall in $G$ (depending on the nonzero eigenvalue of $\varphi_i$ modulo $q$), so the test above works with a minor modification for $u\not\in qV$ as well. Finally, to cover the case $q=2$ and $u\not \in qV$ 
observe that $u\in S$ if and only if $S=Ru$. To test whether this is the case we compute generators for the subgroup 
$H_{Ru}=\{\varphi\in G:\varphi(u)\in Ru\}$ and test membership of these generators for membership in $H$ again by comparing values of the hiding function. 
\end{paragraph}

Equipped with the testing procedure, we compute $S$ from ``bottom u'' as follows. First we compute $S\cap q^{k-1}V$. Note that $q^{k-1}V\cong (\Z/q\Z)^2$ and there are $q+1$ possibilities for $S_{k-1}=S\cap q^{k-1}V$. We can find $S\cap q^{k-1}V$ by brute force based on the testing procedure on all $q+1$ submodules corresponding to each possibility in time $q\poly \log|N|$. Assume that we have computed $S_l=S\cap q^lV$ for
some $l>0$. Then we compute $V_l=\{v\in q^{l-1}V:qv\in S_l\}$ and by an exhaustive search in the factor $V_l/S_l$ we find $S_{l-1}=S\cap q^{l-1}V$ using again the test  in time $q\poly\log|N|$. Note that $V_l/S_l$ is an elementary abelian group of rank at most two. (A factor of a subgroup of an abelian $q$-group generated by 2 elements is also 2-generated.) 
The total cost is $kq\poly\log|N|$. 

We deduce the following result.

\begin{theorem}\label{borel:classical} There is a classical algorithm that solves the Borel hidden subgroup problem in 
$\GL_2(\Z/N\Z)$ in time $\poly(B\log N)$ where $B$ is an upper bound for the largest prime factor of $N$.
\end{theorem}

\subsection{A quantum algorithm}

Denney, Moore and Russell~\cite{DMS10} proposed a quantum polynomial time algorithm that solves the problem in the case when $N$ is a prime. In this subsection we extend their method to arbitrary $N$ as follows. Based on Lemma~\ref{crt} and Theorem~\ref{borel:classical}, it is sufficient to give a procedure that works modulo the odd part of $N$. Thus, in the rest of the discussion we can and do assume that $N$ is odd. 

We use the notation introduced at the beginning of the section. In particular, $H$ is the stabilizer of a secret free cyclic $R$-submodule $S$ of $V=R^2$. Note that $V/S$ is again a free cyclic $R$-module whence for $w=(1,0)^T$ or $w=(0,1)^T$ we have that $S$ and $w$ generate $V$.
We describe an algorithm that works under the assumption that $S$ and $w=(1,0)^T$ generate
$V$. If that fails, we repeat it after an appropriate basis change.

From the assumption, it follows that there is a unique element $s\in R$ such that $S$ is generated by $v=(s,1)^T$.
We restrict the hidden subgroup problem to the stabilizer $K$ of the vector $w$. Note that $K$ is the group consisting of invertible matrices of the form
$$\begin{pmatrix}
1 & * \\
0 & *
\end{pmatrix}.$$ 
Observe that $K$ is isomorphic to the semidirect $R\rtimes R^*$ where the action of $R^*$ on $R$ is the multiplication by its elements and the hidden subgroup $H\cap K$ is the stabilizer of the free cyclic submodule $S$ in $K$. Note that the stabilizer of the submodule $Rv$ in $K$ is the conjugate of the stabilizer of the submodule of $R^2$ generated by $(0,1)^T$ in $K$ by the unitriangular matrix 
$$\begin{pmatrix}
1 & s \\
0 & 1
\end{pmatrix}$$
transporting $(0,1)^T$ to $v=(s, 1)^T$. Hence, by a straightforward calculation, the stabilizer $H\cap K$ of the submodule $Rv$ in $K$ is the subgroup
\begin{equation}
\label{eq:hms}
K_s=\left \{\begin{pmatrix} 1 & hs-s\\ 0 & h
\end{pmatrix}: h\in R^* \right\}.
\end{equation}
As $N$ is odd, the images of the matrices of the form $M-\Id$ where $M\in H\cap K$:
$$\left \{\begin{pmatrix} 0 & (h-1)s\\ 0 & h-1
\end{pmatrix}: h\in R^* \right\}$$ generate the submodule $Rv$.
Indeed, $2\in R^*$ and hence we have $M=\begin{pmatrix} 1 & s\\ 0 & 2 \end{pmatrix}\in H\cap K$ and so the image of $M-\Id$ is $Rv$. 

As shown in the {\em preprint version} of \cite{IPS18}, the HSP in $K\cong R\rtimes R^*$ where the hidden subgroup $H$ is a conjugate of the complement $R^*$ can be cast as an instance of the hidden multiple shift problem $HMS(N,1,r)$ with $r=\phi(N)$, the Euler's totient function of $N$. To see this, note that from (\ref{eq:hms}) it follows that the {\em right} cosets of $K_s$ are of the form $$K_s\cdot \left\{\begin{pmatrix} 1 & a \\ 0 & 1
\end{pmatrix}: h\in R^* \right\}=
\left \{\begin{pmatrix} 1 & hs-s+a\\ a & h
\end{pmatrix}: h\in R^* \right\}.$$ Therefore, when we encode the elements $g$ of $K$ by the second column of $g-\Id$,
the right version of the HSP gives an instance of the hidden multiple shift problem on  $H\times \{h-1:h\in R^*$\}.
As already noted in Section \ref{subsec:HSP}, the left version of the HSP is equivalent to the right one. Therefore, since 
$HMS(N,1,\phi(N))$ can be solved efficiently by Theorem \ref{hms}, 
we obtain the following result. 

\begin{theorem}\label{borel:quantum}
There is a quantum algorithm that efficiently finds the associated free cyclic submodule $S$ for the hidden Borel subgroup $H$ in $\GL_2(\Z/N\Z)$ in time $(\log N)^{O(1)}$.
\end{theorem}

\section{On the isogeny to endomorphism ring problem}
\label{sec: isogeny section} 

In this section, we study the IsERP and its connection to other algorithmic problems. Our final result provides a reduction from the IsERP to a pure quaternion problem, the PQLP (Problem~\ref{problem: powersmooth quaternion lift}), but we obtain this reduction through a quantum equivalence between the IsERP and the Group Action Evaluation Problem, that can be seen as a generalization of Problem~\ref{Prob:endev}.
Most of the work in this section is dedicated to this equivalence. 

In Section~\ref{sec: group action}, we formally introduce  the group action we consider. Then, in Section~\ref{sec: main equivalence}, we prove the result. 
Finally, in Section~\ref{sec: resolution group action problem} we give the link with the PQLP and study the hardness of this problem.

\subsection{The group action of $ \GL_2(\mathbb{Z}/N\mathbb{Z})$ on $N$-isogenies. }
\label{sec: group action}

In this section, we cover all necessary results on the group action we will consider. For that, it is important to understand how $2 \times 2$ matrices $\bmod N$ appear naturally when you consider the action of endomorphisms on the $N$-torsion. 
This comes from the isomorphism $\End(E)/N \End(E) \cong  M_2(\mathbb{Z}/N\mathbb{Z})$ which is a natural extension of the isomorphism $E[N] \cong \Z /N\Z^2 $. We elaborate on that in the next paragraph.

\paragraph*{The isomorphism.}
 Let $P,Q$ be a basis of $E[N]$. We identify any point $R = [x] P + [y] Q$ as the vector $v_R = (x,y)^T$. Then, an endomorphism $\sigma \in \End(E)$ can be seen as a matrix in $M_\sigma \in M_2 (\Z / N \Z) $ through its action on the basis $P,Q$. If we have $\sigma(P) = [a] P + [b]Q$ and $\sigma(Q) = [c]P+ [d] Q$, then we can define $M_\sigma$ as $\begin{pmatrix}
a & c\\
b & d
\end{pmatrix}$ and the representation of $\sigma(R)$ is given by $M_\sigma v_R$. In that case, it can be easily shown that $\det M_\sigma = \deg \sigma \bmod N$. 

If one wants to compute an explicit isomorphism between $\End(E)/N \End(E)$ and $M_2(\mathbb{Z}/N\mathbb{Z})$ one can use the above method of evaluating a basis of $ \End(E)$ on a basis of $E[N]$. However, when $E[N]$ is defined over a large extension field (e.g., $N$ is large random prime), then this method is not efficient. 

The issue can be circumvented by looking at the problem from a slightly different angle. Namely we have basis of $ \End(E)/N \End(E)$ and we also have a multiplication table of the basis elements. Such a representation is called a structure constant representation. Rónyai \cite{ronyai1990computing} proposed a polynomial-time algorithm for this problem when $N$ is prime. The next lemma generalizes the algorithm to arbitrary $N$ whose factorization is known. 
\begin{proposition}\label{prop:expiso}
Let $A$ be a ring isomorphic to $M_2(\mathbb{Z}/N\mathbb{Z})$ given by a structure constant representation. Suppose that factorization of $N$ is known. Then there exists a polynomial-time algorithm that computes an explicit isomorphism between $A$ and $M_2(\mathbb{Z}/N\mathbb{Z})$.
\end{proposition}
\begin{proof}
    First we reduce the problem to the case where $N$ is prime power. Let $N=ab$ where $a$ and $b$ are coprime. Then $A/aA$ is isomorphic $M_2(\mathbb{Z}/a\mathbb{Z})$ and $A/bA$ is isomorphic to $M_2(\mathbb{Z}/b\mathbb{Z})$. Since $M_2(\mathbb{Z}/a\mathbb{Z})\times M_2(\mathbb{Z}/b\mathbb{Z})$ is isomorphic to $M_2(\mathbb{Z}/N\mathbb{Z})$, knowing an explicit isomorphism between $A/a$ and $M_2(\mathbb{Z}/a\mathbb{Z})$ and an explicit isomorphism between $A/bA$ and $M_2(\mathbb{Z}/b\mathbb{Z})$ is enough to recover the isomorphism between $A$ and $M_2(\mathbb{Z}/N\mathbb{Z})$. 
    Using this procedure iteratively (using the fact that the factorization of $N$ is known) one can reduce to the case where $N=q^k$ where $q$ is some prime number. 
    
    Now suppose that $A$ is isomorphic to $M_2(\mathbb{Z}/q^k\mathbb{Z})$. Observe that $A/qA$ is isomorphic to $M_2(\mathbb{Z}/q\mathbb{Z})$. One can compute a non-trivial idempotent in $A/qA$ using Rónyai's algorithm \cite{ronyai1990computing}, let that be $e_0$. Now one has that $e_0^2-e_0\in qA$ and $e_0$ and $e_0-1$ are not in $qA$. Our goal is to find an element $e$ which is an idempotent of $A$. We will perform an iteration which starts with $e_0$ and in the $i$th step we return an element $e_i$ for which $e_i^2-e_i\in q^{i+1}A$. Suppose we have an element $e_{i-1}$ for which $e_{i-1}^2-e_{i-1}\in q^i A$. Now we are looking for an $f\in A$ such that $(e_{i-1}+fq^i)^2-(e_{i-1}+fq^i)\in q^{i+1}A$. This is clearly equivalent to $(e_{i-1}^2-e_{i-1})+(2e_{i-1}-1)fq^i\in q^{i+1}A$. Now dividing by $q^i$ let $(e_{i-1}^2-e_{i-1})/q=E_{i-1}$. We need an $f$ such that that $E_{i-1}+(2e_{i-1}-1)f \in qA$. Observe that since $(e_{i-1}^2-e_{i-1})\in qA$ one has that $(2e_{i-1}-1)^2-1\in qA$, hence $2e_{i-1}$ reduces to an invertible matrix in $A/qA$. Hence any choice for which $f\equiv -E_{i-1}(2e_{i-1}-1)^{-1} \pmod{qA}$ is sufficient and  we can set $e_i=e_{i-1}+fq^i$. Then $e_{k-1}$ will be a non-trivial idempotent of $A$. Non-triviality follows from the fact that every $e_i$ is congruent to $e_0$ modulo $qA$. 
    
    Since $A$ is a $2\times 2$ matrix ring,  $Ae$ is isomorphic to $(\mathbb{Z}/q^k\mathbb{Z})^2$ as an $A$-module. Then the left action of $A$ on $I$ provides an explicit isomorphism between $A$ and $M_2(\mathbb{Z}/q\mathbb{Z})$. We could not find a reference for this fact, so we present a quick simple proof. Let $\begin{pmatrix}
        a&b \\
        c&(1-a)
    \end{pmatrix}$ be an idempotent matrix. We can assume that it has the above form as $e$ is not congruent to 0 or the identity matrix modulo $qA$. We also have that $a(1-a)=bc$. Now the following is a generating set (as an abelian group) of $Ae$: 
    $$ \begin{pmatrix}
        a&b \\
        0&0
    \end{pmatrix},\begin{pmatrix}
        0&0 \\
        a&b
    \end{pmatrix},\begin{pmatrix}
        c&(1-a) \\
        0&0
    \end{pmatrix},
    \begin{pmatrix}
        0&0 \\
        c&(1-a)
    \end{pmatrix}.
    $$
    One has that either $a$ or $1-a$ is invertible in $\mathbb{Z}/q^k\mathbb{Z}$, one may suppose that $a$ is invertible (the calculation is the same in the other case). Then it is clear that every element of the form 
    $$ \begin{pmatrix}
        \alpha&\alpha(b/a) \\
        \beta& \beta (b/a)
    \end{pmatrix}$$
is in the left ideal for any $\alpha,\beta\in \mathbb{Z}/q^k\mathbb{Z}$. We show that every element of $Ae$ is of this form. This follows from the fact that every element of the form     
$$ \begin{pmatrix}
        \gamma c&\gamma(1-a) \\
        \delta c& \delta (1-a)
    \end{pmatrix}$$
    can be written in this form (a linear combination of the second two basis elements) because if $\gamma c=\alpha$ and $\delta c=\beta$, then $\alpha(b/a)=\gamma(1-a)$ and $\beta(b/a)=\delta(1-a)$ (because $cb/a=(1-a)$). Finally, it is clear that the map  $ \begin{pmatrix}
        \alpha&\alpha(b/a) \\
        \beta& \beta (b/a)
    \end{pmatrix}\mapsto (\alpha,\beta)$ is an isomorphism of $A$-modules. 
    \qed 
\end{proof}

\begin{remark}
Once an idempotent $e$ is found, one can finish the proof in an alternate way as well. Namely one can show that $\text{Im} (e)=\ker(e)$ is a cyclic subgroup of $(\Z/q^k\Z)^2$ of cardinality $q^k$. Then a generator of $\ker(e)$ and $\ker(1-e)$ will be a basis in which $e$ is $\begin{pmatrix}
    1&0\\
    0&0
\end{pmatrix}$ which shows indeed that the left ideal generated by $e$ is isomorphic to $\mathbb{Z}/q^k\mathbb{Z}$.
\end{remark}

\paragraph{The group action of invertible matrices on isogenies.}
Now, let us take a cyclic subgroup $G \subset E[N]$ of order $N$ (this is a submodule of rank $1$ inside $(\Z / N \Z)^2$ with our isomorphism). If $\sigma \in  \GL_2(\mathbb{Z}/N\mathbb{Z})$, then it is clear that $\sigma (G)$ is also a cyclic subgroup of order $N$. Thus, we have a natural action of $\GL_2(\mathbb{Z}/N\mathbb{Z})$ on the cyclic subgroups of order $N$. 

This group action on subgroups of order $N$ can naturally be extended to a group action of $\GL_2(\mathbb{Z}/N\mathbb{Z})$ on the set of $N$-isogenies from $E$ through the bijection between cyclic subgroups of order $N$ and $N$-isogenies given by $G \mapsto (\phi_G : E \rightarrow E/G)$ (and whose inverse is simply $\phi \mapsto \ker \phi$).

Composing this bijection with the group action we already have, we get the following group action
\begin{equation}
    \label{eq: group action}
     M_\sigma \star \phi_G \mapsto \phi_{\sigma(G)}.
\end{equation}

This action is always well-defined. However, for computational purposes, we want ways to efficiently represent its elements and compute the action $\star$. These considerations motivate the remaining of this paper.

The problem we consider is the following: 
\begin{problem}[Group Action Evaluation Problem]\label{prob: group action eval}
      Let $E$ be a supersingular elliptic curve over $\mathbb{F}_{p^2}$ and let $\phi: E\rightarrow E_1$ be an isogeny of degree $N$ for some integer $N$.
    Given $\End(E)$, an isogeny representation for $\phi$, $M$ in $\GL_2 (\Z 
 / N \Z)$, find an isogeny representation of $M \star \phi$.      
    % Suppose you can evaluate $\phi$ on any point $P$ of order coprime to $N$. The cost of the evaluation is the size of the representation of $\phi(P)$ (i.e., you can essentially evaluate $\phi$ either on points that are defined over small field extensions or points of powersmooth degree). Compute the endomorphism ring of $E_1$
\end{problem}

\paragraph{The stabilizer subgroups.} One last thing that will be important to apply our results to this group action is to identify the stabilizer subgroup associated to a given isogeny $\phi$. In fact, those are pretty easy to identify and are well-known objects. The answer is given by the following proposition. 

\begin{proposition}
\label{prop: stabilizer subgroup}
Let $\phi : E \rightarrow E'$ be an isogeny of degree $N$. The stabilizer subgroup associated to $\phi$ through the group action defined in Equation~(\ref{eq: group action}) is made of the matrices $M_\sigma$ such that $\sigma$ is in the Eichler order $\Z + I_\phi$ where $I_\phi$ is the ideal associated to $\phi$ under the Deuring correspondence. 
% The order $\Z + I_\phi$ is also equal to $\O_L(I_\phi) \cap \O_R(I_\phi)$
\end{proposition}
\begin{proof}
    By definition of the group action, the stabilizer subgroup is obtained with the matrices $M_\sigma$ such that $\sigma(\ker \phi) = \ker \phi$. This means that $\sigma$ acts as a scalar $\lambda_\sigma$ on $\ker \phi$. Thus, $\ker \phi \subset \ker (\sigma - \lambda_\sigma)$ and by definition of $I_\phi$, we have that $\sigma - \lambda_\sigma \in I_\phi$, hence $\sigma \in \Z + I_\phi$. Conversely, it is clear that any element in $\Z + I_\phi$ acts as a scalar on $\ker \phi$ and so is part of the stabilizer. 
    For the proof that $\Z + I_\phi$ is an Eichler order, see \cite{sqisign}. \qed 
\end{proof}

\begin{remark}
    \label{rmk: endo rings from stabilizer}
    Writing the stabilizer subgroups as Eichler orders of the form $\Z + I_\phi$ will help us prove that
    computing the stabilizer subgroup is essentially equivalent to finding the endomorphism ring of the codomain of $\phi$ (which is isomorphic to the right order of $I_\phi$). 
\end{remark}
\begin{proposition}\label{prop:uptrian}
 The stabilizer subgroups are conjugates of the subgroup of upper triangular matrices (i.e., a Borel subgroup).   
\end{proposition}
\begin{proof}
Follows from \cite[23.1.3]{voight2018quaternion}. For an elementary proof see Appendix \ref{app:elementary}. \qed 
\end{proof}
\subsection{The main reductions.}
\label{sec: main equivalence}

In this section, we prove a quantum polynomial-time equivalence between the the Group Action Evaluation Problem and the IsERP.

\begin{theorem}
    \label{thm: from endo ring to group action}
    The Group Action Evaluation Problem reduces to the IsERP in classical polynomial-time.
\end{theorem}

\begin{proof}
     Assume we have an efficient algorithm to solve the IsERP. 
     Let us take an instance of the Group Action Evaluation Problem. So we have $N,E,\End(E)$, a representation for $\phi$ and a matrix $M$, and we want to compute a representation for $M*\phi$.

    The first step of the reduction is to compute the ideal $I_\phi$ associated to $\phi$. There are several ways to do that, but to keep this proof short, we will use some of the results proven in \cite{leroux2021new}. 
    Thus, our first step is to build a suborder representation for the isogeny $\phi$ as in \cite{leroux2021new}. The suborder representation is made of endomorphisms of $\Z + N \End(E) \hookrightarrow \End(E')$ of powersmooth norm. Since we know $\End(E)$ and $\End(E')$, the algorithms of the Deuring correspondence can be used to compute their kernels in polynomial time (their norm being powersmooth implies that their kernels are defined over a small extension). Then, we can compute the suborder representation using Vélu's formulas.      
     Once we have the suborder representation, we can apply the equivalence between the SOIP and the SOERP \cite[Proposition 13]{leroux2021new} to find the ideal $I_\phi$. Once $I_\phi$ has been computed, we need to compute the ideal $I_{M \star  \phi}$. 
     For that we are going to use a $\sigma$ such that $M_\sigma = M$. 
     We can build such a $\sigma$ in polynomial time from $\End(E)$ using Proposition~\ref{prop:expiso}. 
     
     Once a good $\sigma$ is known, we get the ideal $I_{M \star  \phi}$
     as $ \sigma (I_\phi \cap \O \sigma) \sigma^{-1} + N \O$ (where we take $\O \cong \End(E)$). Since the ideal $I_{M_\sigma \star  \phi}$ is a valid isogeny representation, this proves the result. \qed 
\end{proof}

\begin{theorem}
    \label{thm: quantum reduction}
    The IsERP reduces to the Group Action Evaluation Problem in quantum polynomial time.
\end{theorem}

\begin{proof}
    Assume we can solve the Group Action Evaluation Problem. 

    Let us take an input of the IsERP, we have a curve $E$, its endomorphism ring $\End(E)$, an integer $N$ and the isogeny representation associated to an isogeny $\phi$ of degree $N$. 

    The algorithm to solve the Group Action Evaluation Problem allows us to compute efficiently the group action introduced in Section~\ref{sec: group action}. 
    Using Proposition \ref{prop:uptrian} and Theorem \ref{borel:quantum} one can compute the stabilizer subgroup associated to $\phi$. 
    As the stabilizer subgroups of $\phi$ give us matrices corresponding to some $\sigma$ in the Eichler order $\Z + I_\phi$, we can compute the embedding of this order in $\O \cong \End(E)$ in polynomial time using the algorithm of Proposition~\ref{prop:expiso}. Then, we can extract the ideal $I_\phi$ and compute $O_R(I_\phi)$ which is isomorphic to $\End(E')$ and this gives the result.

\qed
%\begin{remark}
%    \AL{Péter : Is this remark still relevant?}
%    Applying Rónyai's algorithm is a bit of an overkill here as we only have $2\times 2$ matrices here. In that case one also has the following simple algorithm. Represent the algebra as a quaternion algebra and find a zero divisor by finding a zero of a ternary quadratic form over $\mathbb{F}_p$ (reduced norm restricted to the trace zero subspace). Then the action of the algebra on the left ideal generated by the zero divsior will provide an explicit isomorphism.
%\end{remark}

   \end{proof}

    When the degree $N$ is smooth, we can modify the proof of Theorem~\ref{thm: quantum reduction} to get a classical reduction by using Theorem \ref{borel:classical} instead of Theorem \ref{borel:quantum}.  

\begin{theorem}\label{rmk: smooth degree}
    Suppose that degree of the secret isogeny is smooth. Then IsERP reduces to the Group Action Evaluation Problem in classical polynomial time.
\end{theorem}

\subsection{Reduction of the Group Action Evaluation Problem to the PQLP.}
\label{sec: resolution group action problem}
% So far we have established a classical polynomial-time reduction between finding smooth-degree isogenies and evaluating the endomorphism ring group action. In this subsection we study a different problem where the secret isogeny $\phi:\rightarrow E_1$ has a large prime degree. In this setting our goal is to retrieve the endomorphism ring of $E_1$. The motivation for this problem comes from the scheme pSIDH \cite{leroux2021new}. In pSIDH the secret isogeny has a large prime degree and one has access to "unlimited" torsion point information. We provide a more precise formulation of this problem: 

% Our goal is to provide a polynomial-time quantum algorithm that solves Problem \ref{prob:pSIDH} (again assuming a one-to-one correspondence between $f$-isogenous curves to $E$ and cyclic subgroups of order $f$ of $E$). 

In this section, we reduce the Group Action Evaluation Problem to another problem that we call the Powersmooth Quaternion Lift Problem (PQLP).
%
%Then, we give a classical subexponential-time algorithm to solve the PQLP. 
%
The PQLP can be stated as follows: 

\begin{problem}
    \label{problem: powersmooth quaternion lift}
    Let $\O$ be a maximal order in $\QA$. 
    Given an integer $N$ and an element $\sigma_0 \in \O$ such that $(n(\sigma_0),N)=1$, find $\sigma = \lambda\sigma_0 \bmod N \O$ of powersmooth norm with some $\lambda$ coprime to $N$. 
\end{problem}

We use $\textsf{PQLP}_{\mathcal{O}}(\sigma_0)$ to denote the set of $\sigma \in \mathcal{O}$ that satisfy the conditions in Problem~\ref{problem: powersmooth quaternion lift} with respect to $\sigma_0\in \mathcal{O}$. The high level idea of the reduction from the PQLP to the Group Action Evaluation Problem is close to the approach introduced in \cite{kutas2021one}. Given a matrix $M$, the goal is to find a good representative of the class of $M$, i.e. a $\sigma \in \End(E)$ of powersmooth norm, such that $M = M_\sigma$. Then, we can use Vélu's formulae to solve the Group Action Evaluation Problem.

\begin{proposition}
    \label{prop: lifting to group action}
    The Group Action Evaluation Problem reduces to PQLP in classical polynomial time. 
\end{proposition}
\begin{proof}
Let us take an instance of our problem. We have $N, E,  \End(E)$, an isogeny representation for $\phi: E \rightarrow E'$ of degree $N$ and a matrix $M$. 

We need to show that if we know a $\sigma \in \End(E)$ (represented as a quaternion element in a maximal order $\O \cong \End(E)$) of powersmooth norm such that $M_\sigma = M$, then we can compute a representation for $M \star \phi$ in polynomial time. 
For that, we will use the following commutative isogeny diagram
\[
\xymatrix{
E' \ar[r]^{\sigma'} & E'/\sigma(\ker \phi) \\
E\ar[u]^{\phi} \ar[r]^{\sigma} & E \ar[u]_{M \star \phi}
}
\]
where $\sigma'$ has the same degree as $\sigma$ and is defined by $\ker \sigma' = \phi(\ker \sigma)$. 
Since the isogeny $\sigma'$ has powersmooth degree, it can be computed in polynomial time once $\ker \sigma'$ has been computed. Since, we can evaluate $\phi$, it suffices to compute $\ker \sigma$ and this can be done in polynomial-time since we known $\End(E)$. 

Since the diagram is commutative, we have that $\sigma' \circ \phi = M\star \phi \circ \sigma$ and this gives us the way to evaluate efficiently $M \star \phi$ on almost all torsion (as soon as the order is coprime to $\deg \sigma$) as $M \star \phi =  \hat{\sigma'} \circ \phi \circ \sigma/ \deg \sigma$. This is sufficient to build a suborder representation of $M \star \phi$ (see the algorithm outlined in the proof of Theorem~\ref{thm: quantum reduction}).   

This proves that the main computational task is to find this $\sigma$ of powersmooth norm.  
Thus, it suffices to apply an algorithm to solve the PQLP on input $N$, $\End(E)$ and a $\sigma_0$ such that $M_{\sigma_0} = M$ (that we can find using Proposition~\ref{prop:expiso}).  

% Thus, we need to find a good lift of $\sigma_0$, i.e. an element $\omega \in \End(E)$ such that $\sigma_0 + N \omega$ has powersmooth norm. 
\qed
\end{proof}

% The high-level idea is the following. First we show how to compute the endomorphism ring group action utilizing the torsion point information (this is essentially the same idea as in \cite{kutas2021one}). Then we choose an arbitrary isogeny of degree $f$ and compute its codomain $E_2$. This is non-obvious as $f$ is a prime number. Let the kernel of the secret isogeny be $A$ and the kernel of the chosen isogeny be $B$. Then we compute the stabilizer of $A$ and $B$ and show that from that information one can find an endomorphism $\sigma$ that takes $B$ to $A$. Finally, by pushing the kernel of $\sigma$ through our chosen isogeny, we can compute the endomorphism ring of $E_1$ as we know the endomorphism ring of $E_2$. 
% Before moving on to our main theorem we need the following lemma: 

\section{Resolution of the PQLP}\label{sec:respqlp}

% In this section, we resolve PQLP (Problem~\ref{problem: powersmooth quaternion lift}) with conditions on the divisors of $N$ summarized in the following theorem.

% \begin{theorem}\label{thm:pqlpalg}
% Let $N$ be an integer such that $2\nmid N$ and it has $O(\log(\log p))$ many divisors, then PQLP can be solved in polynomial time.
% \end{theorem}

% We postpone the proof to later. The reductions from IsERP to the Group action evaluation problem and then to PQLP show that we have a polynomial time algorithm that resolves IsERP for $N$ satisfies the conditions in Theorem~\ref{thm:pqlpalg}. As an immediate implication, this breaks pSIDH as there $N$ to taken to be a large prime. Despite the restrictions on $N$, the cases we consider for $N$ are viewed as harder cases for IsERP. As discussed earlier, the problem is easy when $N$ is powersmooth, and when $N$ is a product of powersmooth with a few large prime factors, we expect a hybrid approach to work which we breifly discuss in Appendix~\ref{app:hybridPQLP}. Therefore, we essentially have a polynomial time quantuam algorithm that solves IsERP for isogenies of any degree.

In this section, we solve the PQLP (Problem~\ref{problem: powersmooth quaternion lift}) with conditions imposed on the divisors of $N$, as detailed in the following theorem.

\begin{theorem}\label{thm:pqlpalg}
Let $N = \prod \ell_i^{e_i} \neq p$ be an odd integer that is of size polynomial in $p$ and has $O(\log(\log p))$ divisors, then there exists a randomized classical polynomial time algorithm that solves the PQLP.
\end{theorem}

This theorem follows from the correctness of Algorithm~\ref{alg:pqlp} which is introduced and discussed in Section~\ref{sec:pqlpmainalg}. The successive reductions from the IsERP to the Group Action Evaluation Problem, and subsequently to the PQLP, demonstrate the existence of a polynomial time algorithm that solves the IsERP, given $N$ satisfies the conditions in Theorem~\ref{thm:pqlpalg}. As a direct consequence, this breaks pSIDH quantumly since $N$ is a large prime in pSIDH. 
%
%Despite the restrictions on $N$, the chosen $N$ cases are viewed as harder cases for the IsERP. 
As mentioned previously, the IsERP is easy when $N$ is powersmooth. 
We briefly discuss some approaches to solve the general case in Appendix~\ref{app:hybridPQLP}.

%When $N$ is a product of a powersmooth integer with a few large prime factors, we expect a hybrid approach to work which .

In Section~\ref{sec:algorithmic building blocks}, we give a summary of useful known algorithms and provide variants for $\textsf{RepresentInteger}$, $\textsf{StrongApproximation}$ and $\textsf{KLPT}$ to better accommodate our specific application. Following this, we introduce our primary strategy for resolving the PQLP in Section~\ref{sec:pqlpmainalg}. In Section~\ref{sec:quaternion decomposition}, we deal with a critical technical point which we introduce as the Quaternion Decomposition problem. The crux of this problem, and indeed our main conceptual contribution, is the decomposition of $\sigma_0$ into elements that are easy to lift, and elements already possessing a powersmooth norm. Finally in Section~\ref{sec:quantum iserp}, we provide a quantum algorithm that solves the PQLP.

%For the rest of this Section, we assume $N$ to be an integer as in Theorem~\ref{thm:pqlpalg}.

\subsection{Algorithmic building blocks}\label{sec:algorithmic building blocks}

Our algorithm for the PQLP is founded on algorithmic building blocks initially introduced in~\cite{kohel2014quaternion} and later extended in other work, such as~\cite{sqisign}. 
We provide a brief recap of these algorithms here, along with several new variants tailored to suit our requirements. We fix $\log^c p$ to be our powersmooth bound for some constant $c$, and this bound is inherently implied whenever we reference the term `powersmooth'.\CP{last sentence should be moved to some prelims}

As in~\cite{kohel2014quaternion}, for each $p$, we fix a special $p$-extremal maximal order $\mathcal{O}_0$. 
\begin{equation}
\mathcal{O}_0 = 
\begin{cases}
  \ZZ\langle i,\frac{1+j}{2}\rangle \text{ where } i^2 = -1, j^2 = -p,  & \text{for } p \equiv 3 \bmod 4, \\
  %\ZZ\langle i,\frac{1+j+k}{2},\frac{i+2j+k}{4}\rangle \text{ where } i^2 = -2, j^2 = -p ,& \text{for } p \equiv 5 \bmod 8\\
  \ZZ\langle \frac{1+i}{2},j,\frac{ci+k}{q} \rangle \text{ where } i^2 = -q, j^2 = -p, & \text{for } p \equiv 1 \bmod 4,
\end{cases}
\end{equation}
where $c$ is any root of $x^2 + p \bmod q$. 
In the second case, $q$ is required to satisfy that $q \equiv 3 \bmod 4$ is a prime and $\legendre{-p}{q} = 1$. We add one extra condition that $(q,N) = 1$. Under the generalized Riemann hypothesis (GRH), the smallest $q$ is of size $O(\log^2 p)$. 

For each $\mathcal{O}_0$, we identify a suborder of the form $R+Rj$ for $R = \ZZ[i]$ (note that we are making a slightly different choice here than the $R$ in~\cite{kohel2014quaternion} where they always take $R$ to the maximal order in $\QQ(i)$).\CP{Not sure there is a good reason to do that?} For an element $\alpha = a+bi\in R$, we use $\textsf{Re}_R(\alpha)$ to denote $a$ and $\textsf{Im}_R(\alpha)$ to denote $b$. Let $D$ denote the index $[\mathcal{O}_0:R+Rj]$, then
\begin{equation}
    D =  
\begin{cases}
    4, & \text{for } p \equiv 3 \bmod 4, \\
    %8 & \text{for } p \equiv 5 \bmod 8 \\
    4q, & \text{for } p \equiv 1 \bmod 4. \\
\end{cases}
\end{equation}

We will now detail the algorithmic building blocks, sourced from~\cite{kohel2014quaternion} or~\cite{sqisign}.

\begin{itemize}
    \item[--] 
    $\textsf{Cornacchia}(M)$: on an input $M\in \ZZ$, outputs either $\bot$ if $M$ cannot be represented by a fixed quadratic form $f(x,y)$, or a solution $x,y$ to the equation $M = f(x,y)$.
    \item[--] $\mathsf{RepresentInteger}_{\mathcal{O}_0}(M)$: on an input $M>p$, outputs $\gamma \in \mathcal{O}_0$ such that $n(\gamma) = M$.
    \item[--] $\mathsf{StrongApproximation}_F(N,\mu_0)$: on inputs an integer $F>pN^4$, a prime $N$ and $\mu_0\in Rj$, outputs $\lambda \notin N\ZZ$ and $\mu\in \mathcal{O}_0$ of norm dividing $F$ such that $\mu = \lambda \mu_0 \bmod{N\mathcal{O}_0}$.
    \item[--] $\mathsf{KLTP}_M(I)$: on inputs an integer $M>p^3$ and an ideal $I$, outputs an equivalent ideal $J$ such that $n(J) = M$.
\end{itemize}

Let us denote the output $\gamma$ of $\mathsf{RepresentInteger}_{\mathcal{O}_0}(M)$ as $C+Dj$ with $C,D\in R$. To fit our algorithm's specific use case, we require not only that $n(\gamma)$ is powersmooth, but also that $C,D$ satisfy additional conditions relative to some inputs $A,B\in R$, which are determined by $\sigma_0$ from the PQLP. As a result, we introduce the following variant named $\mathsf{RepresentInteger'}_{R+Rj}(N,A,B)$. This variant necessitates more randomized steps to find the desired outputs $C,D \in R$.

%$\mathsf{RepresentInteger'}_{R+Rj}(N,A,B)$. %Given an odd integer $N$ and $A,B\in R$, it outputs $C,D\in R$ such that: i) $(n(C+Dj),N)=1$ and $n(C+Dj)$ is powersmooth; ii) $(n(C),N)=1$; iii) $(n(D),N)=1$; iv) $N\nmid \textsf{Im}(A\Bar{B}C\Bar{D})$ and v) $4q(4p^2n(ABCD)-n(AC)+pn(AD))$ is a square modulo $N$. 
    \begin{algorithm}
	\caption{$\mathsf{RepresentInteger'}_{R+Rj}(N,A,B)$}\label{alg:representinteger'}
    \vspace{.2ex}
    \Input {%
    An integer $N$ and $A,B\in R$}
       
    \Output {%}
	     $C,D\in R$ such that: i) $(n(C+Dj),N)=1$ and $n(C+Dj)$ is powersmooth; ii) $(n(C),N)=1$; iii) $(n(D),N)=1$; iv) $(N, \textsf{Im}_R(A\Bar{B}C\Bar{D}))=1$ and v) $4q(4p^2n(ABCD)-n(AC)+pn(AD))$ is a square modulo $N$.
        }
    \vspace{.4ex}

    Randomly generate $M>p\log^\epsilon p$ with $\epsilon>8$ that is powersmooth and coprime to $N$.  \;

    Set $m =\lfloor \sqrt{\frac{M}{p(q+1)}} \rfloor$ and sample random integers $z,t\in [-m, m]^2$. \label{step:zt}
    
    Set $M' = M - p(z^2+qt^2)$.\;\label{step:M'}
    
    \If{$(z^2+qt^2,N)\neq 1$ or $(M',N)\neq 1$ or $M'$ is not a prime}
    {Go back to Step~\ref{step:zt}.}\label{step:coprime}

    \If{$\textsf{Cornacchia}(M') = \bot $}{Go back to Step~\ref{step:zt}.}

    $x,y = \textsf{Cornacchia}(M')$.\;

    $C\leftarrow x+yi, D \leftarrow z+ti$. \;

    \If{($N, \textsf{Im}_R(A\Bar{B}C\Bar{D}))\neq 1$}{Go back to Step~\ref{step:zt}.}

    \If{$4q(4p^2n(ABCD)-n(AC)+pn(AD))$ is not a square modulo $N$}
    {Go back to Step~\ref{step:zt}.}\label{step:square}
	    \Return $C,D$.
\end{algorithm}
\begin{heuristic}\label{heu:repinteger'}
We assume that $M'$, $z^2+qt^2$, $\textsf{Im}_R(A\Bar{B}C\Bar{D})$ and $4q(4p^2n(ABCD)-(n(AC)-pn(AD))^2)$ appearing in Algorithm~\ref{alg:representinteger'} behave like random integers of the same size. 
%Let $M>p$, and let $m = \lfloor \sqrt{\frac{M}{p(q+1)}} \rfloor$, sampling $z,t\in [-m,m]^2$ randomly
\end{heuristic}
\begin{lemma}\label{lem:repinteger}
Let $N$ be as in Theorem~\ref{thm:pqlpalg}, assuming Heuristic~\ref{heu:repinteger'}, Algorithm~\ref{alg:representinteger'} returns a solution in polynomial time.
\end{lemma}
\begin{proof}
Since $M',z^2+qt^2, \textsf{Im}_R(A\Bar{B}C\Bar{D})$ and $4q(4p^2n(ABCD)-((n(AC)-pn(AD))^2)$ are assumed to behave like random integers, then Step 4 succeeds with probability greater than ${1}/{O(\log^3 p)}$, and Step 10  succeeds with probability greater than ${1}/{O(\log p)}$. 
Since $q$ is of size $O(\log^2 p)$ and $\#\mathcal{C\ell}(\ZZ[i])$ is approximately $O(\log p)$, Step 6 succeeds with probability greater than ${1}/{O(\log p)}$. Finally, Step 12 succeeds with probability greater than ${1}/{O(\log p)}$. In total, the success probability is greater than ${1}/{O(\log^6 p)}$. Since $\epsilon$ is chosen to be greater than 8, there will be enough $(z,t)$ pairs from Step~\ref{step:zt} to ensure Algorithm~\ref{alg:representinteger'} succeeds. Therefore, Algorithm~\ref{alg:representinteger'} will return a solution in polynomial time.
\qed 
\end{proof}
    
We also provide a generalization of the $\mathsf{StrongApproximation}$ algorithm to allow for composite $N$ (note that in this algorithm we don't have restriction on the number of divisors of $N$). The sub-index $_{ps}$ refers to powersmooth.

   \begin{algorithm}
	\caption{$\mathsf{StrongApproximation}_{ps}(N,\mu_0)$}\label{alg:strongapp_powersmooth}
    \vspace{.2ex}
    \Input {%
    An odd integer $N$ and $\mu_0\in Rj$  such that $(n(\mu_0),N)=1$.}
       
    \Output {%}
	     $\lambda \in \ZZ$ such that $(\lambda,N) = 1$ and $\mu \in R$ with powersmooth norm such that $\mu = \lambda \mu_0 \bmod{N\mathcal{O}_0}$.
        }
    \vspace{.4ex}

    Write $\mu_0$ as $(t + si)j$ with $t,s \in \ZZ$.\;

    Let $S = \{\ell \text{ such that } \ell \mid N \text{ and } \legendre{n(\mu_0)}{\ell} = -1\}$.  

    Randomly generate $\#S$ many prime $p_i$'s such that $p_i<\log^c p$, denote $\legendre{p_i}{\ell_j}$ by $\epsilon_{ij}$ for $\ell_j\in S$. \label{step:random primes}

    Solve the system of $\#S\times \#S$ equations $\Sigma_{i=1}^{\#S}\epsilon_{ij}x_{i} = -1$ for $j = 1,\cdots,\#S$ in $\mathbf{F}_2$.\label{step:linearequationmod2}

    %Randomly generate $F>pN^4$ that is powersmooth and coprime to $N$ such that $\frac{F}{p(t^2 + qs^2)}\in \ZZ/N\ZZ$ is a square, and set $\lambda$ to be one of the square root.\CP{may be slightly more enlightening to write $n(\mu_0)$ instead of $p(t^2 + qs^2$?}\;

    \If{There is no solution found in Step~\ref{step:linearequationmod2}}{Go back to Step~\ref{step:random primes}\CP{take all smallest primes and increase the size of $S$ when needed, instead of using random $S$?}}

    $F = \prod_{i=1}^{\#S} p_i^{x_i}$.

    Multiply a $\log^c p$-powersmooth square factor that is coprime to $n(\mu_0)$ to $F$ to ensure $F > pN^4$.\CP{Factor $(q+1)$ missing?}\label{step:finalF}

    Denote one of the square root of $\frac{F}{n(\mu_0)}$ by $\lambda$.\label{step:squarerootlambda}
    
    Randomly generate $c,d$ such that $\lambda p (2tc + 2qsd) \equiv (F - \lambda^2p(t^2+qs^2))/N \bmod N$.\label{step:sample_c_d}

    Set $M = (F - p((\lambda t + cN)^2+q(\lambda s+dN)^2))/N^2$.\label{M=F-()}\;

    \If{$M$ is not a prime or $\textsf{Cornacchia}(M') = \bot $\label{step:Mprime}}{Go back to Step~\ref{step:sample_c_d}}

    $a,b = \textsf{Cornacchia}(M)$.\label{step:alg2Corna}\;

	    \Return $\lambda\mu_0 + N(a+bi+(c+di)j), \lambda$.

     \end{algorithm}

%\CP{Instead of generating $F$ randomly in Step 4, could proceed as follows: select many $q_i$ small primes coprime to $N$; compute symbol $(q_i/\ell_j)$ for all prime factors $\ell_j$ dividing $N$; translate the condition that $Fp(t^2qs^2)$ is a square into a linear system of equations modulo 2; pick a solution with enough factors to ensure $F$ is large enough}

\begin{heuristic}\label{heu:strongapp}
We assume that $M$ appearing in Algorithm~\ref{alg:strongapp_powersmooth} behaves like a random integer of the same size.
\end{heuristic}

\begin{lemma}\label{lem:n/d1d2}
Consider a linear equation $N_1x+N_2y = N_3 \bmod N$ where $\gcd(N,N_1) = d_1, \gcd(N,N_2) = d_2$ and $(d_1,d_2) = 1$. Then this equation has $\frac{N}{d_1d_2}$ solutions in $(\ZZ/N\ZZ)^2$.
\end{lemma}
\begin{proof}
This can be seen by checking how many solutions there are for equation $N_1x+N_2y = N_3 \bmod \ell_i^{e_i}$ with $\ell_i^{e_i}$ being a prime power divisor of $N$ and then using Chinese Reminder Theorem.
\qed
\end{proof}

\begin{lemma}\label{lem:strongapp_powersmooth}
Let $N$ be an odd integer, assuming Heuristic~\ref{heu:strongapp}, Algorithm~\ref{alg:strongapp_powersmooth} returns a solution in polynomial time with overwhelming probability.
\end{lemma}
\begin{proof}
The $\#S\times \#S$ linear equations behave like a random system of linear equations of the same dimension, therefore, repeating Step~\ref{step:random primes} constant number of times will give rise to a linear system over $\mathbf{F}_2$ that is solvable. We have ensured that $F$ generated in Step~\ref{step:finalF} is such that $\frac{F}{n(\mu_0)}$ is a square modulo $N$, hence Step~\ref{step:squarerootlambda} makes sense. 
Similar to what is discussed in the proof of Lemma~\ref{lem:repinteger}, Step~\ref{step:Mprime} succeeds with probability at least $1/O(\log^2 p)$. 
To make sure this algorithm gets passed to Step~\ref{step:alg2Corna}, we need $O(\log^2 p)$ many solutions from Step~\ref{step:sample_c_d}. According to Lemma~\ref{lem:n/d1d2}, this happens if $\gcd(N,t)\gcd(N,s)< N/O(\log^2 p)$. Since $N$ is of size polynomial in $p$, this occurs with overwhelming probability.
\qed 
\end{proof}

Finally, another variant we introduce here is $\mathsf{KLTP}'_M(\mathcal{O}_0,\mathcal{O})$. Here $M$ is an integer such that $M>p^3$. This algorithm first computes a connecting ideal $I'$ from $\mathcal{O}_0$ to $\mathcal{O}$, then computes an equivalent left $\mathcal{O}_0$-ideal $I$ of norm $M$ whose right order is $\alpha^{-1}\mathcal{O}\alpha$ for some nonzero $\alpha \in \mathcal{B}_{p,\infty}$ using $\mathsf{KLPT}_M(I')$. This $\mathsf{KLPT}'$ algorithm outputs $I$ and $\alpha$.

\subsection{The main algorithm}\label{sec:pqlpmainalg}

In this section, we first present a strategy that solves Problem~\ref{problem: powersmooth quaternion lift} for special orders $\mathcal{O}_0$. Then, we expand this strategy to address more general orders $\mathcal{O}$. Let $\sigma_0 \in \mathcal{O}_0$ be as in Problem~\ref{problem: powersmooth quaternion lift}, the method proceeds as follows.

\begin{enumerate}
    \item Find $\sigma'_0\in R + Rj$ such that $\sigma'_0 = \sigma_0 \bmod {N\mathcal{O}_0}$. Since $[\mathcal{O}_0:R+Rj] = D$, $D\sigma_0 \in R+Rj$. Let $D'\in \ZZ$ be such that $D'D \equiv 1 \bmod N$, such $D'$ exists since $(D,N) = 1$, then $\sigma'_0 = D'D\sigma_0\in R+Rj$ and $\sigma'_0 = \sigma_0 \bmod{N\mathcal{O}_0}$. By an abuse of notation, we will use $\sigma_0$ to denote $\sigma'_0$ in what follows. \CP{not clear to me that this step is needed?}
    \item Write $\sigma_0 = A + Bj$ with $A,B\in R$, let $\gamma$ be the output of $\textsf{RepresentInteger}'_{R+Rj}(N,A,B)$. Intuitively, $\gamma$ is an element in $R+Rj$ that has powermooth norm and satisfies extra properties to ensure the next step has a solution.
    
    % Find an element $\gamma\in R+Rj$ such that $n(\gamma)$ is powersmooth. Such an element $\gamma$ can be found efficiently with the $\textsf{RepresentInteger}_{\mathcal{O}_0}$ algorithm \MC{add ref to, Sqisign?}  by feeding it with a powersmooth integer bigger than $p$ that is coprime to $N$, and this integer will be the norm of $\gamma$.
    \item Find $\alpha_1,\alpha_2,\alpha_3\in Rj$ such that $\sigma_0 = \alpha_1\gamma\alpha_2\gamma\alpha_3 \bmod{N\mathcal{O}_0}$. This is the main technical point in this method, we introduce it as Problem~\ref{prob:decompose} and discuss it in detail in Section~\ref{sec:quaternion decomposition}. 
    \item Find $\gamma_i \in \mathcal{O}_0$ such that $\gamma_i = \lambda_i \alpha_i \bmod{N\mathcal{O}_0}$, $n(\gamma_i)$ is powersmooth and $(\lambda_i,N)=1$ for $i = 1,2,3$. These are exactly the outputs of $\mathsf{StrongApproximation}_{ps}(N,\alpha_i)$. 
    \item The element $\gamma_1\gamma\gamma_2\gamma\gamma_3 \in \mathcal{O}_0$ satisfies that $\sigma_0 = \lambda \gamma_1\gamma\gamma_2\gamma\gamma_3 \bmod{N\mathcal{O}_0}$ with $n(\gamma_1\gamma\gamma_2\gamma\gamma_3)$ powersmooth and $\lambda$ coprime to $N$.
\end{enumerate}

In general, let $\mathcal{O}$ be a maximal order in $\mathcal{B}_{p,\infty}$, and let $n_I>p^3$ be a random integer that is coprime to $N$. Let $n_I'\in \ZZ$ such that $n_I'n_I\equiv 1 \bmod N$. Let $I,\alpha$ be the outputs of $\mathsf{KLPT}_{n_I}(\mathcal{O}_0,\mathcal{O})$ such that $I$ is a connecting ideal from $\mathcal{O}_0$ to $\mathcal{O}':=\alpha^{-1}\mathcal{O}\alpha$ and $n(I) = n_I$. We then have inclusions $n_I\mathcal{O}' \subseteq \mathcal{O}_0$ and $n_I\mathcal{O}_0 \subseteq \mathcal{O}'$,  therefore $n_I\alpha^{-1}\sigma_0\alpha \in \mathcal{O}_0$. Let $\sigma \in \textsf{PQLP}_{\mathcal{O}_0}(n_I\alpha^{-1}\sigma_0\alpha)$, then $\sigma = n_I\alpha^{-1}\sigma_0\alpha \bmod{N\mathcal{O}_0}$. Multiplying $n_I$ with both sides of the equation yields that $n_I\sigma = n_I^2\alpha^{-1}\sigma_0\alpha \bmod{N\mathcal{O}'}$. Multiplying $n_I'^2$ on both sides gives that $n_I'\sigma = \alpha^{-1}\sigma_0\alpha \bmod{N\mathcal{O}'}$. Since $n_I'$ is coprime to $N$, $\sigma \in \textsf{PQLP}_{\mathcal{O}'}(\alpha^{-1}\sigma_0\alpha)$, therefore $\alpha\sigma \alpha^{-1} \in \textsf{PQLP}_{\mathcal{O}}(\sigma_0)$. 

We summarize the discussions above and present Algorithm~\ref{alg:pqlp}.

\begin{algorithm}
\caption{\,$\textsf{PQLP}_{\mathcal{O}}(N,\sigma_0) $}\label{alg:pqlp}
    \vspace{.2ex}
    \Input {%
An odd integer $N$ that is of size polynomial in $p$ and has $O(\log\log p)$ distinct prime factors, a maximal order $\mathcal{O}$, and an element $\sigma_0 \in \mathcal{O}$ such that $(n(\sigma_0),N)=1$.}

    \Output {%}
	    $\sigma \in \textsf{PQLP}_{\mathcal{O}}(\sigma_0)$.
        }
    \vspace{.4ex}

   % $\mathcal{O}_0, R, D \leftarrow p$ \;

    Compute $D'$ such that $D'D \equiv 1 \bmod N$.\;

    $\sigma_0 \leftarrow D'D\sigma_0$ \;

    Write $\sigma_0$ as $A+Bj$ with $A,B\in R$\;

    $\gamma \leftarrow \mathsf{RepresentInteger'}_{R+Rj}(N,A,B)$\;

    $\alpha_1,\alpha_2,\alpha_3 \leftarrow \textsf{QuaternionDecomposition} (\sigma_0,\gamma,N)$(Algorithm~\ref{alg:decompose})\;

    $\lambda_i, \gamma_i\leftarrow \mathsf{StrongApproximation}_{ps}(N,\alpha_i) $ for $i = 1,2,3$\;

    $\sigma \leftarrow \gamma_1\gamma\gamma_2\gamma\gamma_3$\;

    Randomly generate $n_I>p^3$ that is coprime to $N$. \; 

    $I,\alpha \leftarrow \mathsf{KLPT}_{n_I}(\mathcal{O}_0,\mathcal{O})$\;

	    \Return $\alpha \gamma_1\gamma\gamma_2\gamma\gamma_3 \alpha^{-1}$.

\end{algorithm}

\subsection{Quaternion decomposition}\label{sec:quaternion decomposition}

In this section, we discuss how to perform Step 3 from the strategy outline. We start with introducing a new problem.

\begin{problem}[Quaternion Decomposition]\label{prob:decompose}
Let $N$ be an odd integer, and $\mathcal{O}_0, R$ be as defined in Section~\ref{sec:algorithmic building blocks}. Let $\sigma_0 = A + Bj,\gamma = C+Dj \in R+Rj$ be such that: i) $(n(\gamma),N)=1$ and $n(\gamma)$ is powersmooth; ii) $(n(C),N)=1$; iii) $(n(D),N)=1$; iv) $(N,\textsf{Im}_R(A\Bar{B}C\Bar{D}))$ and v) $4q(4p^2n(ABCD)-(n(AC)-pn(AD))^2)$ is a square modulo $N$. Find $\alpha_1,\alpha_2,\alpha_3\in Rj$ such that $\sigma_0 = \alpha_1\gamma\alpha_2\gamma\alpha_3 \bmod{N\mathcal{O}_0}$.
\end{problem}

% , then there are two possibilities, either $N$ splits or is inert in $R$. \MC{what happens to ramify?} Let $\textsf{red}: R \rightarrow R/NR$ denote the quotient map. Note that $R/NR \cong \mathbb{F}_{N^2}$ when $N$ is inert and $R/NR \cong \ZZ/N\ZZ \times \ZZ/N\ZZ$ when $N$ splits. Let us state a more general problem that solves the task from the third step.

Suppose one could find $\alpha_1,\alpha_2,\alpha_3 \in Rj$ such that 
\begin{equation}\label{eq:quad_decompose}
\sigma_0 \Bar{\alpha}_3\Bar{\gamma
}= \alpha_1\gamma\alpha_2\bmod{N\mathcal{O}_0},
\end{equation} and $(n(\alpha_3),N) =1 $, then 
$$\sigma_0 = n'_{\alpha_3}n'_\gamma \alpha_1\gamma\alpha_2\gamma\alpha_3.$$
Here $n'_{\alpha_3}$ and $n'_\gamma$ are integers such that $n'_{\alpha_3}n(\alpha_3) \equiv 1 \bmod N$ and $n'_\gamma  n(\gamma) \equiv 1 \bmod N$ respectively. We then search for solutions $\alpha_1,\alpha_2,\alpha_3 \in Rj$ for Equation~(\ref{eq:quad_decompose}) with $(n(\alpha_3),1)=1$ instead.

Let us write $\alpha_i$'s as $x_ij$ with $x_i$'s being unknowns that we wish to find in $R$, writing Equation~(\ref{eq:quad_decompose}) in terms of $A,B,C,D,x_1,x_2$ and $x_3$ gives rise to the following:

\begin{align*}
 \text{Equation~(\ref{eq:quad_decompose})} \iff (A + Bj)(-j)\Bar{x}_3(\Bar{C}-j\Bar{D}) & = x_1j(C+Dj)x_2j\bmod{N\mathcal{O}_0}\\  
\iff (-Ax_3j+pB\Bar{x}_3)(\Bar{C}-j\Bar{D}) &= (x_1\Bar{C}j-px_1\Bar{D})x_2j\bmod{N\mathcal{O}_0}\\
\iff (-pA\Bar{D}x_3 + pB\Bar{C}\Bar{x}_3) + (-ACx_3 - pBD\Bar{x}_3)j &= -p\Bar{C}x_1\Bar{x}_2 - p\Bar{D}x_1x_2j \bmod{N\mathcal{O}_0}.
\end{align*}

Therefore, in order to solve the original equation, it suffices to find $x_1,x_2,x_3 \in R$ with $(n(x_3),N)=1$ such that 

\begin{equation}\label{eq:eq1}
    \begin{cases}
       pA\Bar{D}x_3 - pB\Bar{C}\Bar{x}_3 = p\Bar{C}x_1\Bar{x}_2 \bmod{NR},  \\
        ACx_3 + pBD\Bar{x}_3 = p\Bar{D}x_1x_2 \bmod{NR}.
    \end{cases}
\end{equation}
Note that the modulo condition in Equation~(\ref{eq:eq1}) holds not just in $N\mathcal{O}_0$ but in $NR$ since $[\mathcal{O}_K:R]=1$ or $2$ is coprime to $N$. By assumption, $n(C)$ and $n(D)$ are both coprime to $N$, let $n'_C, n'_D\in \ZZ$ be integers such that $n'_Cn(C)\equiv 1 \bmod N$ and $n'_Dn(D)\equiv 1 \bmod N$ respectively, and let $p'\in \ZZ$ be such that $p'p\equiv 1 \bmod N$, then Equation~(\ref{eq:eq1}) is equivalent to 

\begin{equation}\label{eq:eq2}
    \begin{cases}
      (n'_Cp') (pA\Bar{D}x_3 - pB\Bar{C}\Bar{x}_3) = x_1\Bar{x}_2 \bmod{NR},  \\
        (n'_Dp')(ACx_3 + pBD\Bar{x}_3) = x_1x_2 \bmod{NR}.
    \end{cases}
\end{equation}

\begin{lemma}\label{lem:samenorm}
Equation~(\ref{eq:eq2}) has a solution $(x_1,x_2,x_3)\in R^3$ if and only if there exists $x_3 \in R$ such that $x:=(n'_Cp') (pA\Bar{D}x_3 - pB\Bar{C}\Bar{x}_3)$ and $y: = (n'_Dp')(ACx_3 + pBD\Bar{x}_3)$ have same norm modulo $N$.
\end{lemma} 

\begin{proof}
One solution to Equation~(\ref{eq:eq2}) clearly implies $n(x)=n(y)=n(x_1)n(x_2)$. For the other direction, we provide a simple proof here for the case when $N$ is a prime that is inert in $R$, and refer to Appendix~\ref{app:norm1} for the general case when $N$ is an arbitrary odd integer. Since $N$ is a prime that is inert in $R$, $R/(N) \cong \mathbb{F}_{N^2}$. Hilbert's Theorem 90 implies that if $a \in R/(N)$ has norm 1, then $a = b/\Bar{b}$ for $b \in R/(N)$. Since $n(x) = n(y)$ and both $x,y \notin NR$, we have that $n(x/y)=1$, therefore $x/y = z/\Bar{z}$ for some nonzero $z\in R/(N)$. Consequently, $x_1,x_2$ can be chosen to be lifts of $yz$ and $1/z$ to $R$ respectively.
\qed
\end{proof}

\begin{remark}
The method we present for odd integers $N$ is constructive, therefore leads to an algorithm that finds $x_1,x_2$ given $x_3$ such that $n(x)=n(y)$. We call this algoithm $\textsf{EquivNormConjugationProduct}(x_3)$.
\end{remark}

% \begin{proof}
% Clearly if a solution $(x_1,x_2,x_3)\in R^3$ is found, then one has $n(x) = n(y) \bmod N$. Now let us assume that we found $x_3\in R$ such that $n(x) \equiv n(y) \bmod N$.  Write $x = u_1+v_1i$ and $y = u_2 + v_2i$, we divide the discussion into two cases:
% \begin{itemize}
%     \item[Case I:] At least one of the four integers $u_1+u_1,u_1-u_2,v_1+v_2$ and $v_1-v_2$ is coprime to $N$. Without loss of generality, let us assume $(v_1+v_2)$ and we use $v'$ to denote an integer such that $v'(v_1+v_2)\equiv 1 \bmod N$, then $(v'(u_1+u_2)+i,\zeta(v_1+v_2)+\zeta q'(u_1-u_2)i,x_3)$ is a solution where $\zeta,q'$ are integers such that $\zeta2\equiv1\bmod N$ and $q'q\equiv 1 \bmod N$ respectively. One could verify that 
%     \begin{align*}
%         (v'(u_1+u_2)+i)(\zeta(v_1+v_2)-\zeta q'(u_1-u_2)i) = u_1+\zeta(v_1+v_2-q'v'(u_1^2-u_2^2)) = u_1+v_1i = x \bmod {NR}, \\
%         (v'(u_1+u_2)+i)(\zeta(v_1+v_2)+\zeta q'(u_1-u_2)i) = u_2 + \zeta(v_1+v_2+q'v'(u_1^2-u_2^2)) = u_2+v_2i = y \bmod {NR}.
%     \end{align*}
%     The formulae for the solutions in other cases can be worked out similarly, see Algorithm~\ref{alg:decompose} for more details.
%     \item[Case II:] None of these four integers $u_1+u_1,u_1-u_2,v_1+v_2$ and $v_1-v_2$ is coprime to $N$. Let us denote their gcd with $N$ by $d_1,d_2,d_3$ and $d_4$ respectively. Since \[N\mid (u_1+u_2)(u_1-u_2)+q(v_1+v_2)(v_1-v_2),\] it has to be the case that either $d_1\mid d_3$ or $d_3\mid d_1$
% \end{itemize}
% \end{proof}

The condition $n(x) = n(y)$ is equivalent to $n(CDpx) = n(CDpy)$. And one could calculate explicitly that 
\begin{align*}
    n(CDpx) &= p^2n(A)n^2(D)n(x_3) + p^2n(B)n(C)n(D)n(x_3)- 2p^2n(D)\textsf{Re}_R(A\Bar{B}C\Bar{D}x_3^2),\\
    n(CDpy) &= n(A)n^2(C)n(x_3)+p^2n(B)n(C)n(D)n(x_3)+2pn(C)\textsf{Re}_R(A\Bar{B}C\Bar{D}x_3^2).
\end{align*}

%\CP{Since those elements are not really in $\mathbb{C}$, maybe use scalar product $\langle\alpha,\beta\rangle=\alpha\bar\beta\beta\bar\alpha$ instead of $Re$ and $Im$? That is, define $t:=\langle A\bar D,B\bar C\rangle=\langle AC,BD\rangle$}\MC{changed}\MC{I ended up changing them back but modified the $\textsf{Re}$ and $\textsf{Im}$ notations with a subindex and define them when introducing $R$. Otherwise, $\textsf{Im}$ looks too bad with inner product}

We now aim to find $x_3 = s + ti \in R$ with $(n(x_3),N)=1$ such that $n(CDpx) - n(CDpy) = n(\gamma)\Bigl(n(A)(n(C)-pn(D))+2p\textsf{Re}_R(A\Bar{B}C\Bar{D}x_3^2)\Bigr) = 0 \bmod N$. Plugging in $x_3 = s+ti$, finding $x_3$ is equivalent to finding $(s,t)\in \ZZ^2$ such that 
\begin{equation}\label{eq:binary_quad_x3}
   f(s,t) := C_1s^2 +C_2st + C_3t^2 = 0\bmod N,
\end{equation}
where $C_1 = \Bigl(n(A)\bigl(n(C)-pn(D)\bigr)+2p\textsf{Re}_R(A\Bar{B}C\Bar{D})\Bigr), C_2 = - 4qp(\textsf{Im}_RA\Bar{B}C\Bar{D})$ and $C_3 = \Bigl(qn(A)\bigl(n(C)-pn(D)\bigr)-2qp\textsf{Re}_R(A\Bar{B}C\Bar{D})\Bigr)$. 
Clearly, a solution $(s,t)$ exists if and only if the discriminant 
\begin{equation}\label{eq:discriminant}
4q(4p^2n(ABCD)-(n(AC)-pn(AD))^2)
\end{equation}
%\CP{This is $$4qn(A)[n(BCD)-n(A)(n(\gamma))^2]$$}\MC{should be $+p$ instead $-p$ to be $n(\gamma)$?}
%\CP{Some typo? cofactor of $Im^2$ in $C_2^2$ is $16q^2p^2$; cofactor of $Re^2$ in $-4C_1C_3$ is $16qp^2$ (i.e. missing a factor $q$)?}\MC{I noticed this as I calculated, this seems to be correct, I will double check}\MC{I did use the exact formula here for $C_1,C_2,C_3$ in the implementation to calculate $x_3$}\CP{Existence of a decomposition should not change if we premultiply $\sigma_0$ by $E\in R$ (either on the left or the right), because $E$ could be absorbed in either $\alpha_1$ or $\alpha_3$. However, discriminant condition above seems affected by such a change?}\MC{makes sense... I will try to see what's going on...}\MC{turns out I forgot a square, now it doesn't change by multiplying elements from $R$, thanks!}
of the quadratic equation is a square modulo $N$. 
 By our assumption this is the case. Then $(s,t)$ viewed in $(\ZZ/N\ZZ)^2$ is defined up to a scalar, and we could simply choose $s = 1$ and let $t_0$ be one of the root of $f(1,t)\equiv 0\bmod N$. 

\CP{What does that condition become if you use two different $\gamma$'s in the decomposition? (i.e. fix one of them; what are the conditions on the second one, are they simpler than the condition above?) What if you take $\gamma$ and $\bar\gamma$?}

Finally, suppose $(n(x_3),N) = 1$ does not hold, this implies $(N, \textsf{Im}_R(A\Bar{B}C\Bar{D}))$ which contradicts our assumption. Therefore, we have shown that we could find a solution $x_3 = s+ti$ where $s,t$ satisfies Equation~(\ref{eq:binary_quad_x3}) and $(n(x_3),1)=1$. 

We now summarize our algorithm for solving Problem~\ref{prob:decompose} in Algorithm~\ref{alg:decompose}. 

\begin{algorithm}
	\caption{\,$\textsf{QuaternionDecomposition}$(N,A,B,C,D)}\label{alg:decompose}
    \vspace{.2ex}
    \Input {%
$N,A,B,C,D$ as in Problem~\ref{prob:decompose}.
}

    \Output {%}
	    $x_1,x_2,x_3\in R$ such that $A+Bj = \lambda x_1j(C+Dj)x_2j(C+Dj)x_3j \bmod{N\mathcal{O}_0}$ where $\lambda$ is some integer that is coprime to $N$.
        }
    \vspace{.4ex}

    $t_0\leftarrow $ root of $\Bigl(n(A)\bigl(n(C)-pn(D)\bigr)+2p\textsf{Re}_R(A\Bar{B}C\Bar{D})\Bigr) - 4qp\textsf{Im}_R(A\Bar{B}C\Bar{D})t + \Bigl(qn(A)\bigl(n(C)-pn(D)\bigr)-2qp\textsf{Re}_R(A\Bar{B}C\Bar{D})\Bigr)t^2 = 0\bmod N.$

    $x_3 \leftarrow 1 + t_0i$\;

    $x_1,x_2 \leftarrow \textsf{EquivNormConjugationProduct}(x_3)$ \;
    
    \Return $x_1,x_2,x_3$.

\end{algorithm}
% We summarize our results in the following theorem:
% \begin{theorem}\label{thm:pqlpalg}
% Let $N$ be an odd number with at $O(\log\log N)$ prime factors. Let $O$ be a maximal order in $B_{p,\infty}$ and let $\sigma\in O$. Then there there is a randomized polynomial-time algorithm that finds $\tau\in O$ such that the norm of $\sigma+N\tau$ is powersmooth. \MC{I assumed $N$ to be of size polynomial in $p$ to state everything in terms of $\log p$} 
% \end{theorem}
% \begin{proof}
%  \PK{do we need a proof here or is it clear from the previous discussions} \MC{I don't think we need a proof and we have a simmilar theorem in the beginning of this section}  
% \end{proof}
\subsection{Quantum algorithm for the PQLP}\label{sec:quantum iserp}

As discussed earlier, Theorem~\ref{thm:pqlpalg} implies that we can solve the IsERP in quantum polynomial time. However, Theorem \ref{thm:pqlpalg} only guarantees a randomized polynomial-time algorithm, and it might be advantageous to avoid that inside a quantum algorithm. So instead of lifting elements inside the quantum algorithm we lift $O(\log N)$ elements first and then utilize them to make the lifting procedure inside the quantum algorithm deterministic (and free of any heuristic after the precomputation has succeeded).

\begin{theorem}\label{thm: pqlp with precomputation}
 There is an algorithm that solves the PQLP in quantum polynomial time. % with classical polynomial-time precomputation. 
\end{theorem}
\begin{proof}

We provide a proof for the case where $N$ is a prime number. The proof for general $N$ is in Appendix \ref{app:general N}. For any matrix $M$ in $\GL_2(\mathbb{Z}/N \Z)$, $M$ can be written as $PL\cdot D\cdot U$ where $P$ is a permutation matrix, $L$ is a lower unitriangular (it is lower triangular with 1-s in the diagonal), $D$ is diagonal, and $U$ is upper unitriangular (it is upper triangular with 1-s in the diagonal). This decomposition can be found in polynomial time using Gaussian elimination. Now one has to decompose $L,D$ and $U$ separately. 
Any lower unitriangular matrix can be written as a power of $A=\begin{pmatrix}
    1&0 \\
    g&1
\end{pmatrix}$ where $g$ is a generator of $(\mathbb{Z}/N\mathbb{Z})^*$. Similarly, every upper unitriangular matrix can be written as a power of $B=\begin{pmatrix}
    1&g \\
    0&1
\end{pmatrix}$. Any diagonal matrix can be written as $C^kD^l$ where 
$$C=\begin{pmatrix}
    g&0 \\
    0&1 
\end{pmatrix},D=\begin{pmatrix}
    1&0 \\
    0&g
\end{pmatrix} .$$
This shows that every element in $\GL_2(\mathbb{Z}/N \Z)$ can be written as $PA^aB^bC^cD^d$. Thus instead of lifting elements inside the main quantum algorithm using Theorem \ref{thm:pqlpalg} one can precompute lifts of $P,A,B,C,D$ and then decompose a matrix $M\in \GL_2(\mathbb{Z}/N \Z)$ as $A^aB^bC^cD^d$ to obtain a lift of $M$. 

This decomposition requires several instances of the discrete logarithm problem in  $\mathbb{Z}/N \Z$ which can be computed in quantum polynomial time. The only issue is that if one computes a powersmooth lift of $A$ (or $B,C$ or $D$) then $A^a$ is not going to be powersmooth if $a$ is large. To circumvent this issue one also computes lifts of $A^{2^k},B^{2^k},C^{2^k}$ and $D^{2^k}$ for every $k$ between 1 and $\log_2(N)$. Furthermore, one computes lifts which are coprime (this can be ensured easily as $\textsf{StrongApproximation}$ can lift an element in $\mathbb{Z}[i]j$ to any number that is bigger than $p^{O(1)}$ \PK{not sure what to put here in the exponent for the general case}). 

This way $PA^aB^bC^cD^d$ will also be powersmooth as it is the product of $4\log N+1$ powersmooth numbers. 
%both form cyclic subgroups, so one can lift generators for them in precomputation. Diagonal matrices form a subgroup isomorphic to $(\mathbb{Z}/N\mathbb{Z})^*\times \mathbb{Z}/N\mathbb{Z})^*$, so again one can lift generators in precomputation. Then righting $L$ and $U$ as a power of those generators is a discrete logarithm problem, so can be done in quantum polynomial time. Similarly, one can also do that for $D$ in the context of a generalized discrete logarithm problem (where the group has rank 2 as a $\mathbb{Z}$-module). The final issue is that raising a powersmooth number to a large power will result in a number that is smooth but no longer powersmooth. To circumvent this issue one needs to lift $2^k$-powers of these generators where $k$ is between 1 and $\log_2(N)$. One also has to be careful to ensure that the lifts are close to coprime. \AL{Péter: that last fact requires a bit more details in my opinion.}
Lifting $4\log N+1$ numbers can be done in classical polynomial time using Theorem \ref{thm:pqlpalg}. 

% Our second step is the evaluation of the endomorphism ring group action. Let $\phi: E\rightarrow E_1$ be the secret isogeny whose kernel is $A$. Now given $E_1\cong E/A$ and $\sigma\in End(E)/fEnd(E)$ one has to compute $E/\sigma(A)$. Now the main observation here is that $E/\sigma(A)$ is isomorphic to $E_1/\phi(\ker(\sigma))$ similarly as in \cite{kutas2021one}. Now since we can evaluate $\phi$ on any point we can compute $E_1/\phi(\ker(\sigma))$ if the kernel generator of $\sigma$ can be represented efficiently. For a general $\sigma\in End(E)$ this is not the case, however, note that $\sigma$ is only determined modulo $fEnd(E)$. Thus in order to compute $E_1/\phi(\ker(\sigma))$ we choose a representative of the coset of $\sigma$ that has a powersmooth norm. 
%This is a special case of strong approximation and can be solved with a variant of the KLPT algorithm \cite{kohel2014quaternion}. 
\qed 
\end{proof}

%\paragraph*{More efficient algorithms to solve the PQLP.}

%In fact, the PQLP has been already studied in the context of isogeny-based cryptography. Indeed, an efficient algorithm to solve it for some special case of $\sigma_0$ is at the heart of the KLPT algorithm introduced in \cite{kohel2014quaternion}. More specifically, the article presents a heuristic polynomial time algorithm to solve the PQLP when $\O$ is a special order that contains a quadratic order $R$ of small discriminant and $\sigma_0 \in R$. We can generalize this algorithm to treat cases where $\sigma$ is contained in a quadratic order orthogonal to $R$, but it seems hard to go outside these few cases as the algorithm in \cite{kohel2014quaternion} heavily uses the fact that $\sigma_0 \in R$, and that Cornacchia's algorithm allows us to solve norm equations in $R$ efficiently. Unfortunately, the few nice cases we know how to solve efficiently, do not seem enough to perform the precomputation required by Theorem~\ref{thm: pqlp with precomputation}, even if we assume that $\O$ contains a nice $R$. 
%Since solving equations in quadratic orders of big discriminant is hard already, it seems complicated to extend the result to a generic maximal order $\O$ and any $\sigma_0$.

%Interestingly, an efficient algorithm to solve the PQLP for any $\O$ and $\sigma_0$ could have constructive applications for isogeny-based cryptography by improving the KLPT algorithm, which is also why we think that this problem is worth studying.  

\section{Impact on isogeny-based cryptography}\label{sec:conclusion}

% Our  results can be concisely captured in the following theorem.

% \begin{theorem}\label{thm:solvingIsERP}
%    Let $N = \prod \ell_i^{e_i} \neq p$ be an odd integer that is of size polynomial in $p$ and has $O(\log(\log p))$ divisors. Then there exists a quantum polynomial-time algorithm that solves the IsERP.
% \end{theorem}

% We believe that our method for solving the IsERP can be extended to a wider class of $N$ values, as we discuss in Appendix \ref{app:hybridPQLP}. 

The most important application of Theorem \ref{thm:solvingIsERP}  is that it breaks pSIDH in quantum polynomial time as $N$ is a prime number in pSIDH. 
Another application is on SCALLOP \cite{de2023scallop}. Even though Theorem \ref{thm:solvingIsERP} does not break SCALLOP, it shows that its security reduces to the problem of evaluating the secret prime degree isogeny (up to a scalar). In \cite{de2023scallop} it is already discussed that one can deduce some information on the secret isogeny $\phi$ by utilizing the fact that one can evaluate $\phi\circ\theta\circ\hat{\phi}$ efficiently on any point where $\theta$ is some fixed endomorphism on a curve which has an endomorphism of low degree (typically that curve is $j$-1728 and $\theta$ is the non-trivial automorphism). 

Our results mildly generalize to the following setting. Let $E$ be a supersingular elliptic curve that does not possess a non-scalar endomorphism of degree $N^2$. In other words there is a one-to-one correspondence between cyclic subgroups of order $N$ and $N$-isogenous curves. Our results imply that if given some curve $E/A$ and an endomorphism $\sigma$ one can compute $E/\sigma(A)$, then one can also compute the endomorphism ring of $E/A$ in quantum polynomial time (and actually the corresponding isogeny itself its degree is smooth). The difference between our previous results here is that we do not need an isogeny representation for the isogeny corresponding to the subgroup generated by $A$ as long as we can evaluate the above group action. At the moment we do not see any particular application for this observation but it might prove to be a useful cryptanalysis tool in the future. 
\subsubsection{Acknowledgements} Gábor Ivanyos is supported in part by the Hungarian Ministry
of Innovation and Technology NRDI Office within the framework of the Artificial
Intelligence National Laboratory Program. Péter Kutas is supported by the Hungarian Ministry
of Innovation and Technology NRDI Office within the framework of the Quantum Information National Laboratory Program,
the J´anos Bolyai Research Scholarship of the Hungarian Academy of Sciences and
by the UNKP-22-5 New National Excellence Program. Mingjie Chen, Péter Kutas and Christophe Petit are partly supported by EPSRC through grant number EP/V011324/1.

\bibliographystyle{alpha}
\bibliography{references}
\appendix
\section{Proof of Theorem \ref{thm: pqlp with precomputation} for arbitrary $N$}\label{app:general N}

The main difficulty in the general case is that there is no $PLU$ decomposition in $\GL_2(\Z/N\Z)$. A simple example is the case where $N=6$ and the matrix we need to decompose is 
$$\begin{pmatrix}
    2&3\\
    3&2
\end{pmatrix}.$$ 
This happens because the first column of an invertible matrix (in this case with determinant 1) may  contain no invertible element.

However, if one takes a matrix whose top left corner is an invertible element, then it can be written as a product of a lower and an upper triangular matrix (which then can be further decomposed as in Theorem \ref{thm: pqlp with precomputation}). 

Let $\begin{pmatrix}
    a&b\\
    c&d
\end{pmatrix}$ be an arbitrary matrix in $\GL_2(\Z/N\Z)$. Since it is invertible there exists $k\in\Z/N\Z$ such that $a+kc$ is invertible and such a $k$ can be computed efficiently (using Chinese remaindering). This implies that $\begin{pmatrix}
    1&k \\
    0&1
\end{pmatrix}\begin{pmatrix}
    a&b\\
    c&d
\end{pmatrix}$ has an invertible top left corner (and is invertible as a matrix) hence can be written as $LDU$. This implies that $\begin{pmatrix}
    a&b\\
    c&d
\end{pmatrix}$ can be written as $U_1LDU_2$ where $U_1,U_2$ are upper unitriangular matrices. Finally, for general $N$ the group of lower and upper triangular matrices are no longer cyclic but one can compute generators of $\Z/N\Z$ an then everything else from the proof carries over. 
\section{Elementary proof that the stabilizers are Borel subgroups}\label{app:elementary}
 \begin{proposition}
 Let $X$ be a point of order $N$ in $E$ and let us denote the stabilizer of $E/X$  by $G_X$. Then $G_X$ is a conjugate of the group of upper triangular matrices. 
 \end{proposition}
\begin{proof}
 Let $P,Q$ be some basis of $E[N]$. Then $X=\alpha P+\beta Q$ and there exists integers $x,y$ such that $\alpha x+\beta y\equiv 1 \pmod{N}$, otherwise $X$ would have lower order. One has the following simple identity:
 $$\begin{pmatrix}
 \alpha&-y\\
 \beta &x
 \end{pmatrix}\cdot (1,0)^T=(\alpha,\beta)^T.$$
 Now the condition that $\sigma(A)=\lambda A$ implies that 
 $$\begin{pmatrix}
 a&b \\
 c&d
 \end{pmatrix} 
\cdot
\begin{pmatrix}
 \alpha&-y\\
\beta &x
 \end{pmatrix}\cdot (1,0)^T=\lambda \begin{pmatrix}
\alpha&-y\\
 \beta &x
 \end{pmatrix}\cdot (1,0)^T.$$
 Now one has that 
 $$\begin{pmatrix}
 \alpha&-y\\
 \beta &x
 \end{pmatrix}^{-1}
 \cdot
 \begin{pmatrix}
 a&b \\
 c&d
 \end{pmatrix} 
 \cdot
 \begin{pmatrix}
 \alpha&-y\\
 \beta &x
 \end{pmatrix}\cdot (1,0)^T=\lambda \cdot (1,0)^T.$$
 Since the middle matrix runs through conjugation is an automorphism one has to look at the problem of finding the matrices $\begin{pmatrix}
 x&y \\
 z&u
 \end{pmatrix} $ for which $\begin{pmatrix}
 x&y \\
 z&u
 \end{pmatrix}\cdot (1,0)^T=\lambda (1,0)^T .$
 From comparing coefficients one gets that $z=0$ and all those matrices also satisfy this condition. This proves that $G_A$ is just the the conjugate of upper triangular matrices by the matrix $\begin{pmatrix}
\alpha&-y\\
\beta &x
 \end{pmatrix}$.  
 \qed
\end{proof}

\section{Full proof of Lemma~\ref{lem:samenorm}}\label{app:norm1}

We provide a full proof for Lemma~\ref{lem:samenorm}.

\begin{proof}
Let us write $x_1$ as $t_1 + s_1i$, $x_2$ as $t_2 + s_2i$, $x$ as $u_1 + v_1i$ and $y$ as $u_2 + v_2i$ for $t_1,s_1,t_2,s_t,u_1,v_1,u_2,v_2 \in \ZZ$. Then Equation~(\ref{eq:eq2}) is equivalent to 
\begin{equation}\label{eq:eq5}
    \begin{cases}
       %(t_1+s_1\omega)(t_2-s_2\omega) = %u_1 + v_1\omega \bmod{NR}, \\
       (t_1t_2 + qs_1s_2) + (-t_1s_2+s_1t_2)i= u_1 + v_1i \bmod{NR}, \\
       %(t_1+s_1\omega)(t_2+s_2\omega)  = %u_2 + v_2\omega \bmod{NR}.\\
        (t_1t_2 - qs_1s_2) + (t_1s_2+s_1t_2)i = u_2 + v_2i \bmod{NR}.
    \end{cases}
\end{equation}
Let $\zeta$ be an integer such that $2\zeta \equiv 1 \bmod N$, and let $M_1 = \zeta(u_1+u_2),M_2 = \zeta q'(u_1-u_2),M_3 = \zeta(v_1+v_2)$ and $M_4 = \zeta(v_2-v_1)$. Equation~(\ref{eq:eq5}) is equivalent to the following:

\begin{equation}\label{eq:eq9}
    \begin{cases}
       t_1t_2 = M_1 \bmod{N}, \\
       s_1s_2  = M_2 \bmod{N},\\
       s_1t_2 = M_3 \bmod{N},\\
       t_1s_2= M_4 \bmod{N}.    
    \end{cases}
\end{equation}
Since $N = \prod_{i=1}^{i=k}\ell_i^{e_i}$, by Chinese Reminder Theorem, Equation~(\ref{eq:eq9}) has a solution $(t_1,s_1,t_2,s_2)$ in $\ZZ^4$ if and only if the following equation has a solution in $\ZZ^4$ for each $i$. 

\begin{equation}\label{eq:modelli}
    \begin{cases}
       t_1t_2 = M_1 \bmod{\ell_i^{e_i}}, \\
       s_1s_2  = M_2 \bmod{\ell_i^{e_i}},\\
       s_1t_2 = M_3 \bmod{\ell_i^{e_i}},\\
       t_1s_2= M_4 \bmod{\ell_i^{e_i}}.    
    \end{cases}
\end{equation}
We claim that Equation~(\ref{eq:modelli}) has a solution in $\ZZ^4$ if and only if $M_1M_4-M_2M_3\equiv 0\bmod{\ell_i^{e_i}}$. We prove this claim by considering the following three cases:
\begin{enumerate}
    \item At least one of $M_j$'s is invertible in $\ZZ/\ell_i^{e_i}$, and we assume its $M_1$ without loss of generality. Then $(t_1,s_1,t_2,s_2) = (1,\frac{M_3}{M_1},M_1,M_4)$ is a solution.
    \item None of $M_j$'s is invertible but at least one of $M_j$'s is nonzero in $\ZZ/\ell_i^{e_i}$. Let $e = \text{min}_{j=1,2,3,4}\{\text{val}_{\ell_i}(M_j)\}$, reducing Equation~(\ref{eq:modelli}) to $\ZZ/\ell_i^{e_i-e}$ will give us a solution $(t_1,s_1,t_2,s_2)$ in $\ZZ/\ell_i^{e_i-e}$ since we are back to the previous case. Thinking $(t_1,s_1,t_2,s_2)$ as in $\ZZ$ and $(p^et_1,p^es_1,t_2,s_2)$ is a solution of Equation~(\ref{eq:modelli}).
    \item $M_1,M_2,M_3,M_4$ are all zero in $\ZZ/\ell_i^{e_i}\ZZ$, then $(0,0,0,0)$ is a solution.
\end{enumerate}
The condition $M_1M_4-M_2M_3\equiv 0\bmod{\ell_i^{e_i}}$ holds for each $i$ if and only if $M_1M_4-M_2M_3\equiv 0\bmod{N}$. Plugging in the values of $M_j$'s shows that this is equivalent to $n(x) \equiv n(y)\bmod N$. Therefore, in this case, lifting the solutions modulo each $\ell_i^{e_i}$ will give rise to a solution of Equation~(\ref{eq:eq2}).
\qed 
\end{proof}

\section{IsERP for general $N$}\label{app:hybridPQLP}

In the main body of the paper we restricted ourselves to $N$ with not too many prime factors due to the fact that if $N$ has too many prime factors, then a random number is unlikely to be a square modulo $N$. 
Such a quadratic residuosity condition appears in our strategy to solve the PQLP, more precisely for the discriminant~\ref{eq:discriminant}.

Here we present several approaches that could remove the restriction on $N$.

\subsection{Torsion-point attacks}

The worst case for our algorithm is when $N$ is smooth and has many different prime factors. This is however exactly the case where classical torsion-point attacks are useful. For simplicity we assume that the starting curve is $y^2=x^3+x$ but the general approach should follow in a similar fashion. Let $M$ be a divisor of $N$ of size $poly(\log p)$. Let the secret isogeny be $\phi: E_1\rightarrow E_2$. First we try to find an endomorphism $\theta\in End(E_1)$ and an integer $d$ such that the degree of $\psi=\phi\circ\theta\circ\hat{\phi}+d$ is a powersmooth number $B$. Since we know how $\phi$ acts on the $B$-torsion we can compute $\psi$. Now the usual step here would be to compute $Ker(\psi-d)\cap E_2[N]$ but the $N$-torsion is not defined over a small extension field. Instead we decompose $\phi=\phi_1\circ\phi_2$ where the degree of $\phi_1$ is $M$. Now instead we compute  $Ker(\psi-d)\cap E_2[M]$ which can be accomplished in polynomial time as $M$ is powersmooth. 
This intersection will be the kernel of $\hat{\phi_1}$ (modulo some technical detailed already handled in \cite{petit2017faster}). Thus we can compute $\hat{\phi_1}$ and translate the isogeny representation and we now have a smaller $N$ to work with. We can do this iteratively to handle smooth (and not just powersmooth) divisors of $N$ and this also works for powers of 2 (which is not handled in the main body of the paper as $N$ is assumed to be odd). 

This does not handle all cases as for example $N\approx 2^n$ could have $\sqrt{n}$ prime divisors all of size $2^{\sqrt{n}}$.

\subsection{Fixing $n(C),n(D)$ modulo $N$}

A completely different direction could be to use special $\gamma$s to ensure that $4q(4p^2n(ABCD)-(n(AC)-pn(AD))^2)$ is a square. The idea here is that this discriminant only depends on $n(C)$ and $n(D)$ modulo $N$ and is also homogeneous, hence it is enough to specify $n(C)$ and $n(D)$ modulo $N$ up to an invertible scalar. This is more less a problem of finding a powersmooth number in a residue class, thus a purely elementary number theory problem \PK{i might think about this a bit more whether we can solve this. If not, I will leave it like this. }

\subsection{Improved lifting idea}

Our most promising idea to  handle the general case is the following one.

Assume that $N$ has $k$ different prime factors. Then a residue modulo $N$ is a square if and only if it is a square modulo every divisor of $N$. A random number has thus $1/2^k$ chance of being a square. However, a random number is usually a square modulo half of the prime divisors of $N$ (assuming that $N$ is odd). Let $\sigma_0$ be the quaternion we would like to lift. Thus one can do the first step in the lifting procedure which will result in a $\sigma_1$ which is not perfect lift but $\sigma_0 \equiv \sigma_1 \bmod N_1O$ where $N_1$ divides $N$ and $N_1\approx\sqrt{N}$. The lifting procedure ensures that $\sigma_1$ is invertible modulo $N$ so $\sigma_0/\sigma_1\equiv 1 \bmod N_1O$. This Now we have reduced the original problem to lifting special elements, namely ones that are congruent to $1$ modulo $N_1O$. If we can lift such an element, then we can multiply it with $\sigma_1$ (which is powersmooth) and get a lift of $\sigma_0$. 

Now for elements congruent to 1 modulo $N_1O$ we will try to write them as a product $\alpha_1\gamma\alpha_2\gamma\alpha_2$ where every component is in $\mathbb{Z}+N_1O$ (as we only care about lifting modulo a scalar). We can also assume that the $\alpha_i\in R$ as those elements can be lifted similarly to elements in $Rj$ (elements in $R$ can be written as the product of two elements in $Rj$ but also a more direct lifting procedure is provided in \cite{kutas2021one}). When doing this lifting condition modulo $N$ we run into a similar discriminant issue but now the mod $N_1$ part is automatically satisfied. So then one now expects that the quadratic residue condition is satisfied for half of the prime divisors of $N/N_1$. So thus we obtain a lift os $\sigma_0/\sigma_1$ modulo $N_1N_2$ where $N_2$ divides $N/N_1$ and $N_2\approx \sqrt{N/N_1}$. We can thus iterate this procedure and in $O(\log N) $ steps we potentially arrive at a complete lift. Then the product of $O(\log(N))$ powersmooth numbers is still powersmooth.

\end{document}